\documentclass[12pt,draftclsnofoot,onecolumn]{IEEEtran}
\usepackage{setspace}
\renewcommand{\baselinestretch}{1.5}

\makeatletter
\newcommand\semihuge{\@setfontsize\semihuge{21.5}{22}}
\makeatother

\usepackage{cite,graphicx,epstopdf,color,amssymb,amsopn,amsthm,array,multirow,eqparbox,flushend,hhline,tablefootnote,amsfonts,dsfont,threeparttable}
\usepackage{setspace}
\usepackage{algorithm,algorithmicx,algpseudocode}
\makeatletter
\newcommand{\algmargin}{\the\ALG@thistlm}
\makeatother
\algnewcommand{\parState}[1]{\State%
  \parbox[t]{\dimexpr\linewidth-\algmargin}{\strut #1\strut}}

\usepackage{lipsum,graphicx,multicol}

\usepackage[cmex10]{amsmath}

\usepackage[justification=centering]{caption}
\allowdisplaybreaks
\usepackage{csquotes}

\usepackage{fancyhdr}

\usepackage{soul}

\usepackage{tikz}
\usetikzlibrary{shapes}

\usepackage{verbatim}

\newtheorem{theorem}{\bf Theorem}

\newtheorem{definition}{\bf Definition}

\usepackage{soul}

\begin{document}
\clearpage
\title{\semihuge Evolutionary Games for Correlation-Aware Clustering in Massive Machine-to-Machine Networks \vspace{-0.3cm}}
\author{Nicole~Sawyer,~\IEEEmembership{Student Member,~IEEE,},~Mehdi Naderi Soorki,Walid Saad~\IEEEmembership{Senior Member,~IEEE,}
        and~David~B.~Smith,~\IEEEmembership{Member,~IEEE.}

\thanks{N. Sawyer and D. B. Smith are with the Research School of Engineering, The Australian National University, Canberra ACT 0200, Australia.  The authors are also with CSIRO Data61, Canberra ACT 2601, Australia, e-mail:{nicole.sawyer,david.smith}@data61.csiro.au.}
\thanks{M. Naderi Soorki and W. Saad are with the department of Electrical and Computer Engineering, Virginia Tech, Blacksburg, VA, USA, e-mail: {mehdin,walids}@vt.edu.}
\thanks{This research is supported by an Australian Government Research Training Program (RTP) Scholarship, by the U.S. Office of Naval Research (ONR) under Grant N00014-15-1-2709.}
\thanks{A preliminary version of this work appeared in IEEE GLOBECOM 2017 \cite{sawyer2017evoluationary}.}}

\maketitle\vspace{-0.8cm}
\vspace{-0.7cm}
\thispagestyle{empty}
\begin{abstract}\vspace{-0.4cm}	
In this paper, the problem of self-organizing, correlation-aware clustering is studied for a dense network of machine-type devices (MTDs) deployed over a cellular network.  In dense machine-to-machine networks, MTDs are typically located within close proximity and will gather correlated data, and, thus, clustering MTDs based on data correlation will lead to a decrease in the number of redundant bits transmitted to the base station.  To analyze this clustering problem, a novel utility function that captures the average MTD transmission power per cluster is derived, as a function of the MTD location, cluster size, and inter-cluster interference.  Then, the clustering problem is formulated as an evolutionary game, which models the interactions among the massive number of MTDs, in order to decrease MTD transmission power.  To solve this game, a distributed algorithm is proposed to allow the infinite number of MTDs to autonomously form clusters.  It is shown that the proposed distributed algorithm converges to an evolutionary stable strategy (ESS), that is robust to a small portion of MTDs deviating from the stable cluster formation at convergence.  The maximum fraction of MTDs that can deviate from the ESS, while still maintaining a stable cluster formation is derived.  Simulation results show that the proposed approach can effectively cluster MTDs with highly correlated data, which, in turn, enables those MTDs to eliminate a large number of redundant bits.  The results show that, on average, using the proposed approach yields reductions of up to $23.4$\% and $9.6$\% in terms of the transmit power per cluster, compared to forming clusters with the maximum possible size and uniformly selecting a cluster size, respectively.
\end{abstract}
\vspace{0.15cm}
{\small \emph{Index Terms}--- Data-centric clustering; Data correlation; Evolutionary stable strategy; M2M communications.}

\section{Introduction}
Machine-to-machine (M2M) communications is an important component of the emerging Internet of Things (IoT) system, as it enables advanced networked applications such as smart home technologies, smart grid, healthcare, drone systems, manufacturing systems, and surveillance~\cite{chen2014machine,zheng2012radio,verma2016machine,mozaffari2016unmanned,chen2017caching,park2016learning}.  Within an M2M network, a massive number of machine-type devices (MTDs) will be densely deployed over wireless cellular networks~\cite{dawy2017toward}.  An MTD can be a sensor, actuator, or smart meter whose typical role is to sense or measure an environment, and transmit the collected data to cellular base stations (BSs).  MTDs enable real-time monitoring and control of any physical environment, without direct human involvement, thus making processes more efficient and improving human welfare~\cite{verma2016machine,lien2011toward}.  Since the number of MTDs is expected to be massive and much larger than the number of cellular-type devices (CTDs), it is expected that a massive-scale MTD deployment will lead to quality-of-service (QoS) degradation, increased traffic and signaling, increased latency, and reduced reliability~\cite{verma2016machine}.  Therefore, deploying MTDs over cellular networks faces many challenges ranging from network modeling to resource management, massive-scale, access, and MTD clustering as mentioned in~\cite{chen2014machine,dawy2017toward,verma2016machine}, and~\cite{abuzainab2017cognitive}.

Recently, the idea of clustering MTDs into smaller groups has emerged as a promising technique to reduce the traffic load on the cellular BS and improve spatial reuse and energy efficiency, while reducing interference in the network, as studied in~\cite{juan2013data,lien2011toward,zheng2012radio}, and~\cite{song2015correlation}.  Existing clustering techniques for M2M communications~\cite{miao20162,lee2013feasibility,soorki2017stochastic,ho2012energy,al2016self,wei2012joint,hussain2014multi,safdar2013resource,lien2011toward,juan2013data,song2015correlation,sawyer2017evoluationary,bayat2014distributed} have focused on clustering MTDs based on resource allocation, location, load on the random access channel (RACH), and data correlation.  Clustering has been considered in literature, as an effective approach to alleviate the potential massive congestion caused by MTDs, as done in~\cite{miao20162,lee2013feasibility,soorki2017stochastic} and~\cite{ho2012energy}.  These aforementioned works aim to maximize the number of MTDs that attempt to simultaneously access the BS, while minimizing network congestion, the load on the RACH and signaling overhead~\cite{lee2013feasibility,soorki2017stochastic}.  In~\cite{miao20162}, an energy-efficient cluster formation (load adaptive multiple access scheme) and cluster head selection scheme was proposed, to maximize network lifetime in a massive M2M network.  In~\cite{lee2013feasibility}, a cognitive M2M communication system for a large number of MTDs was considered, in order to reduce congestion on the RACH, by reducing the number of MTDs directly accessing the BS.  Furthermore, the work in~\cite{soorki2017stochastic}, investigated the problem of random access contention between cooperative groups of MTDs that coordinate their random access channel, while taking into account energy consumption and time varying queue length.  On the other hand, clustering techniques based on the QoS requirements and locations of MTDs, are proposed in~\cite{lien2011toward},~\cite{al2016self,wei2012joint,hussain2014multi,safdar2013resource} and~\cite{bayat2014distributed}, in order to maximize the number of supported MTDs.  A cluster prioritization scheme for massive access management is studied in~\cite{lien2011toward}, where MTDs are clustered based on QoS requirements.  In~\cite{abuzainab2017cognitive}, a distributed resource (time) allocation scheme was proposed, to address the diverse QoS requirements in an IoT network, while taking into account data rate of CTDs and energy consumption of MTDs.  The work in~\cite{al2016self} proposes a self-organized cluster formation mechanism in which MTDs form clusters with neighboring MTDs.  Additionally, the work in~\cite{wei2012joint} considers cluster formation based on location to the designated cluster head and a QoS threshold, while minimizing interference in the M2M network.  A number of works such as in~\cite{wei2012joint,hussain2014multi}, and~\cite{safdar2013resource} have also considered joint clustering and resource allocation (such as power control).  The goal of these works is to maximize MTD data rate and minimize MTD transmit power, in order to reduce interference to the cellular network and other MTDs, as well as prolonging MTD battery lifetime.  The aforementioned works consider a machine-centric clustering approach, that clusters MTDs in order to maximize MTD data rate and the number of supported MTDs in the network, while minimizing energy consumption.  Meanwhile, in~\cite{saad2009distributed}, the authors proposed a game-theoretic clustering algorithm that optimizes the tradeoff between sum-rate gains and power costs.

Due to the dense deployment of MTDs in M2M networks, MTDs within close proximity will gather correlated data, thus often sending the same information (redundant bits) to the BS~\cite{hsieh2015not,chang2012not,song2015correlation}.  Hence, a data-centric clustering approach can be used to improve the data quality sent to the BS. Existing work on clustering MTDs with respect to data correlation remains limited.  Primarily, the works in~\cite{juan2013data} and~\cite{song2015correlation} have studied the possibility of MTD clustering based on location and correlation, however, these works rely on centralized approaches that are not practical for large-scale M2M networks.  Such centralized clustering approaches can lead to significant signaling overhead as they require gathering of global information, such as location and data correlation factors, for a large number of MTDs.  Moreover, in practice, centralized clustering approaches are not robust to small changes in the MTD networking environment, such as the joining of new MTDs, MTD loss of battery, or rapid fluctuations in the sensing environment.  Meanwhile, the work in~\cite{hamidouche2017popular} proposed a distributed correlation-aware cell association algorithm, that maximizes the information sent to the small BS, as well as maximizing the number of assigned IoT devices to every small BS.  Even though some of the proposed works, such as,~\cite{soorki2017stochastic} and~\cite{hamidouche2017popular}, consider a distributed approach, they cannot be applied to massive M2M networks.  This is due to the fact that previous works such as~\cite{soorki2017stochastic} and~\cite{hamidouche2017popular} do not consider a massive number of MTDs, as this would cause the aforementioned distributed algorithms to significantly increase in complexity.  Indeed, existing works in~\cite{miao20162,lee2013feasibility,soorki2017stochastic,ho2012energy,wei2012joint,hussain2014multi,al2016self, safdar2013resource,lien2011toward,juan2013data,song2015correlation,bayat2014distributed}, and~\cite{hamidouche2017popular} consider clustering a small, finite number of MTDs, which is not the case for practical IoT scenarios in which the number of MTDs within the network is massive and nearly infinite.  Additionally, the aforementioned works are also not robust to the dynamics of a large-scale M2M network due to the factors such as the arrival or departure of new MTDs or the deactivation of MTDs (e.g., due to battery loss).  Therefore, a distributed clustering approach for a massive number of MTDs is needed, while ensuring low signaling overhead and robustness for small changes in the network.

The main contribution of this paper is to address the problem of fully distributed correlation-aware clustering for an infinite number of MTDs densely deployed in a cellular network.  In particular, we formulate the MTD clustering problem as a dynamic evolutionary game.  In this game, MTDs can self-organize in a fully distributed manner to form clusters, based on data correlation and potential transmission power savings.  Then, using stochastic geometry, we accurately model and characterize the distance distributions between MTDs which, in turn, allows deriving a closed-form expression for the inter-cluster interference, i.e., the interference generated from MTDs in other clusters to a cluster head.  Based on the distance distributions of MTDs and inter-cluster interference, for the proposed evolutionary game, we derive a closed-form expression for the utility function per cluster, as a function of MTD distance distributions, inter-cluster interference, and cluster size.  Hence, by combining both game theory, for distributed decision making, and stochastic geometry, for deriving concrete utility functions, we analyze the evolutionary stable state (ESS) (i.e., the equilibrium) of the proposed game.  In particular, we propose a novel, distributed algorithm that enables the MTDs to autonomously form clusters and converge to an ESS cluster formation, that is robust to potential changes in the cluster membership decisions of the MTDs that can occur due to the joining of new MTDs, MTD loss of battery, or rapid fluctuations in the sensing environment.  In this regard, we also derive the maximum portion of MTDs that can deviate from the ESS, while still maintaining a stable cluster formation.  Simulation results show how the proposed dynamic evolutionary-based correlation-aware clustering algorithm converges to an ESS.  The results also show that, the proposed evolutionary algorithm can substantially decrease transmission power per cluster, and increases the number of redundant bits that can be eliminated in a given cluster, compared to a pure cluster type baseline and a uniform cluster type baseline.  In addition, the results show that, as the network density and data correlation increase, the proposed distributed algorithm determines the number of MTDs within each cluster, based on network density and cluster radius.  Moreover, the simulation results uncover a tradeoff between cluster size and maximum transmission power per cluster, given the network density and correlation constant.  In summary, the novelty of our contributions is outlined as follows:
\begin{itemize}
  \item We characterize the energy consumption and radius of an MTD cluster.  We also investigate the impacts of MTD density, data correlation of MTDs, and the inter-cluster interference, when clustering MTDs.  Additionally, based on stochastic geometry, a closed-form expression for inter-cluster interference is also determined.
  \item Based on the derived analysis results from stochastic geometry, a distributed evolutionary game-theoretic approach is proposed to cluster MTDs with correlated data, in order to decrease MTD transmission power by reducing the number of redundant bits, in a massive and locally finite M2M network.
  \item A fully distributed algorithm based on an evolutionary game is proposed to find the stable cluster formation of MTDs.  The robustness of the proposed distributed algorithm, with respect to the maximum number of MTDs that can change their cluster formation, is derived.  Moreover, the accuracy of the stochastic geometry analysis as well as the effectiveness of the game-theoretic approach are corroborated by extensive simulations.
\end{itemize}

The rest of the paper is organized as follows.  In Section~\ref{Sec:Sys-Model}, we define the system model and problem formulation for correlation-aware clustering of MTDs.  In Section~\ref{sec:EGgame}, an evolutionary game based on data correlation and transmission energy is proposed to cluster an infinite number of MTDs.  As well as using stochastic geometry to characterize distances and inter-cluster interference.  Within Section~\ref{Sec:Stable}, the proposed evolutionary game stability is analyzed.  Simulation results are illustrated in Section~\ref{Sec:Sim}.  Finally, conclusions are drawn in Section~\ref{Sec:Con}.

\section{System Model And Problem Formulation}\label{Sec:Sys-Model}
Consider the uplink of a wireless cellular network having a single BS serving an infinite number of MTDs engaged in M2M communications.  All MTDs in the set $\mathcal{M}$ are randomly deployed in a 2D space $\mathds{R}^2$, where the location of each MTD $m$ is denoted by $\boldsymbol{y}_m,~\forall m\in\mathcal{M}$, and is based on a Poisson point process (PPP) $\Phi_{M}=\{\boldsymbol{y}_m\in\mathds{R}^2|m\in\mathcal{M}\}$ with density $\lambda_m$.  We assume that MTDs gather data from a Gaussian random field, as this model can capture the worst-case scenario for the number of bits needed to encode a source field~\cite{hsieh2015not,chang2012not}.  Thus, the data source $s_m$ for each MTD $m$, is a Gaussian random variable with mean $\mu_m$ and variance $\sigma_m^2$~\cite{hsieh2015not}.  Here, we use entropy to model the information of each MTD's data.  Each MTD quantizes its continuous Gaussian data source with a sufficiently small quantization step $\Delta$.  Hence, the entropy $H_m$ of each MTD's quantized data source will be given by~\cite{hsieh2015not}:
\begin{equation}\label{entropy}
H_m=\frac{1}{2}\log_2\left(\frac{2\pi e}{\Delta^2}\sigma_m^2\right).
\end{equation}

In this network, each MTD $m$ sends its data via a cellular link to the BS.  The BS allocates orthogonal resource blocks to each MTD for the cellular link, based on an orthogonal frequency-division multiple access (OFDMA) system.  A set $\mathcal{Z}$ of $Z$ resource blocks are used for cellular link transmission, with fixed bandwidth $B$~Hz per resource block, during each time slot $t$ with a fixed duration $T$.  Each MTD $m$ is allocated one resource block $z\in\mathcal{Z}$, and a transmit power $p_m\in[0,P_{\max}]$, for cellular link transmission.  The transmission power per resource block $z\in\mathcal{A}$ that MTD $m$ requires to send $H_m$ bits over the cellular link during each time slot, is given by:
\begin{equation}\label{cellular_TxP}
  p_m(t) = \left(\frac{BN_0}{g_{m}(t)^{(z)}}\right)\left(2^{\frac{H_m}{TB}}-1\right),
\end{equation}
where $g_{m}(t)^{(z)}$ is the channel gain of MTD $m$ over the cellular link on resource block $z$; and $N_0$ is the noise power spectral density.

In the uplink of existing cellular networks, the number of available orthogonal resource blocks have been designed to suit the needs of CTDs.  Thus, the number of available resource blocks for each MTD is limited, and as a result, not all MTDs will be allocated orthogonal resources for transmission.  If MTDs are unable to send their sensed data to the BS, this may lead to missing information for particular a environment.  However, due to the dense deployment of MTDs, MTDs within close proximity will gather correlated data due to sensing the same environment~\cite{hsieh2015not}.  For example, sensors deployed to measure temperature in the same room, will measure very similar readings, which is highly correlated.  Hence, multiple MTDs will potentially send a large number of redundant bits (same information) to the BS.  Next, we introduce a \emph{correlation-aware clustering scheme} that will enable an infinite number of MTDs to cooperatively form clusters by sharing their data via M2M links.  Clustering allows a reduction in the number of MTDs accessing the cellular uplink, thus saving orthogonal resources in the uplink for CTDs.  Additionally, within each cluster, MTDs can coordinate and eliminate the redundant bits of the shared data, which means that the overall number of bits sent via the cellular link can be reduced, thus minimizing transmission energy and improving overall system efficiency.

\subsection{Problem Formulation}
The goal of clustering in a dense M2M network is to decrease the number of redundant bits sents to the BS, and to decrease transmission power on the cellular link, as well as to decrease the number of MTDs transmitting to the BS.  Within each cluster, an MTD is designated as a \emph{cluster head} and will relay all of the data gathered from its cluster to the BS over the cellular link.  All cluster members send their data via M2M links to the designated cluster head.  Thus, the BS needs to allocate one resource block per cluster for the cellular link, instead of one resource block per MTD.  Moreover, clustering MTDs based on data correlation also reduces MTD transmission power consumption, and thus, energy-efficiency is increased and MTD battery lifetime is prolonged.  We assume that orthogonal resource blocks are allocated to the M2M links in each cluster.  The resource blocks of cellular links are orthogonal to M2M links.  Thus, each MTD $m$ is allocated one resource block, and transmit power per resource block is bounded by $q_m\in[0,Q_{\max}]$, for M2M link transmission.  Hence, there is no interference between any M2M links in the same cluster, that is, no intra-cluster interference.  However, M2M links in different clusters can simultaneously reuse the same resource blocks, which causes interference to M2M links in other clusters, that is, inter-cluster interference.

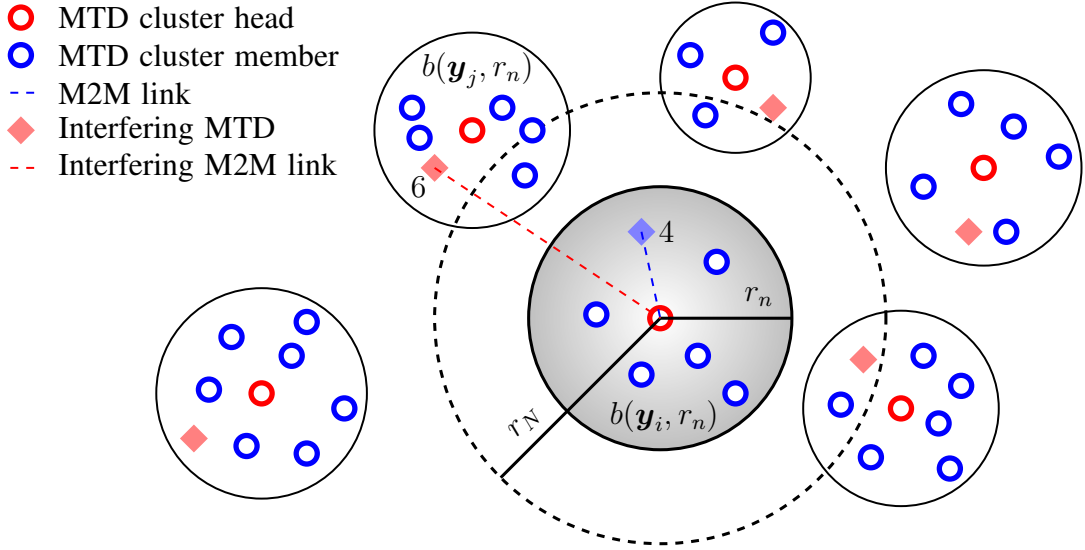
\begin{figure}
\centering
\begin{tikzpicture}
\draw[style=dashed,line width=1.15pt] (0,0) circle (3);
\shadedraw[inner color=white, outer color=gray!50, draw=gray] (0,0) circle (1.75);
\draw[line width = 1.15pt] (0,0) circle (1.75);
\node[text width=1.15pt] at (-0.65,-1.3) {$b(\boldsymbol{y}_i,r_n)$};
\draw[line width=2pt, color=red, fill=white] (0,0) circle (0.14);
\node[draw,scale=0.55,diamond,line width=1.5pt, color=blue!50, fill=blue!50] at (-0.25,1.15){};
\draw[dashed,line width=0.75pt,color=blue] (-0.25,1.15) -- (0,0);
\node[text width=1pt] at (0,1.15) {$4$};
\draw[line width=2pt, color=blue, fill=white] (0.75,0.75) circle (0.14);
\draw[line width=2pt, color=blue, fill=white] (-0.85,0.05) circle (0.14);
\draw[line width=2pt, color=blue, fill=white] (-0.25,-0.75) circle (0.14);
\draw[line width=2pt, color=blue, fill=white] (1,-1) circle (0.14);
\draw[line width=2pt, color=blue, fill=white] (0.5,-0.5) circle (0.14);
\draw[line width=1.15pt] (0,0) -- node[above,near end]{$r_{n}$} (0:1.75);
\draw[line width=1.15pt] (0,0) -- node[near end,above,sloped]{$r_{N}$} (225:3);
\draw[line width=0.75pt] (4.3,2) circle (1.3);
\draw[line width=2pt, color=red, fill=white] (4.3,2) circle (0.14);
\node[draw,scale=0.55,diamond,line width=1.5pt, color=red!50, fill=red!50] at (4.1,1.15){};
\draw[line width=2pt, color=blue, fill=white] (4.6,1.15) circle (0.14);
\draw[line width=2pt, color=blue, fill=white] (4,2.85) circle (0.14);
\draw[line width=2pt, color=blue, fill=white] (3.5,1.75) circle (0.14);
\draw[line width=2pt, color=blue, fill=white] (5.3,2.15) circle (0.14);
\draw[line width=2pt, color=blue, fill=white] (4.7,2.55) circle (0.14);
\draw[line width=0.75pt] (1,3.2) circle (1);
\draw[line width=2pt, color=red, fill=white] (1,3.2) circle (0.14);
\node[draw,scale=0.55,diamond,line width=1.5pt, color=red!50, fill=red!50] at (1.5,2.8){};
\draw[line width=2pt, color=blue, fill=white] (1.5,3.8) circle (0.14);
\draw[line width=2pt, color=blue, fill=white] (0.6,2.7) circle (0.14);
\draw[line width=2pt, color=blue, fill=white] (0.4,3.5) circle (0.14);
\draw[line width=0.75pt] (-2.5,2.5) circle (1.3);
\draw[line width=2pt, color=red, fill=white] (-2.5,2.5) circle (0.14);
\node[text width=1.15pt] at (-3.15,3.3) {$b(\boldsymbol{y}_j,r_n)$};
\node[draw,scale=0.55,diamond,line width=1.5pt, color=red!50, fill=red!50] at (-3,2){};
\draw[dashed,line width=0.75pt,color=red] (-3,2) -- (0,0);
\node[text width=1pt] at (-3.3,1.75) {$6$};
\draw[line width=2pt, color=blue, fill=white] (-1.8,1.9) circle (0.14);
\draw[line width=2pt, color=blue, fill=white] (-1.7,2.5) circle (0.14);
\draw[line width=2pt, color=blue, fill=white] (-3.3,2.8) circle (0.14);
\draw[line width=2pt, color=blue, fill=white] (-3.2,2.4) circle (0.14);
\draw[line width=2pt, color=blue, fill=white] (-2.1,2.8) circle (0.14);
\draw[line width=0.75pt] (-5.3,-1) circle (1.4);
\draw[line width=2pt, color=red, fill=white] (-5.3,-1) circle (0.14);
\node[draw,scale=0.55,diamond,line width=1.5pt, color=red!50, fill=red!50] at (-6.2,-1.6){};
\draw[line width=2pt, color=blue, fill=white] (-4.7,-0.05) circle (0.14);
\draw[line width=2pt, color=blue, fill=white] (-5.5,-1.7) circle (0.14);
\draw[line width=2pt, color=blue, fill=white] (-4.2,-1.2) circle (0.14);
\draw[line width=2pt, color=blue, fill=white] (-4.9,-0.5) circle (0.14);
\draw[line width=2pt, color=blue, fill=white] (-5.7,-0.25) circle (0.14);
\draw[line width=2pt, color=blue, fill=white] (-6,-0.95) circle (0.14);
\draw[line width=2pt, color=blue, fill=white] (-4.7,-1.8) circle (0.14);
\draw[line width=0.75pt] (3.2,-1.2) circle (1.3);
\draw[line width=2pt, color=red, fill=white] (3.2,-1.2) circle (0.14);
\node[draw,scale=0.55,diamond,line width=1.5pt, color=red!50, fill=red!50] at (2.7,-0.55){};
\draw[line width=2pt, color=blue, fill=white] (3.7,-1.4) circle (0.14);
\draw[line width=2pt, color=blue, fill=white] (2.4,-1.15) circle (0.14);
\draw[line width=2pt, color=blue, fill=white] (4,-0.9) circle (0.14);
\draw[line width=2pt, color=blue, fill=white] (3.5,-0.5) circle (0.14);
\draw[line width=2pt, color=blue, fill=white] (3.85,-2) circle (0.14);
\draw[line width=2pt, color=blue, fill=white] (2.8,-1.85) circle (0.14);
\draw[line width=2pt, color=red, fill=white] (-8.5,4) circle (0.14);
\draw[line width=2pt, color=blue, fill=white] (-8.5,3.5) circle (0.14);
\draw[dashed,line width=0.75pt,color=blue] (-8.64,3) -- (-8.3,3);
\node[draw,scale=0.55,diamond,line width=1.5pt, color=red!50, fill=red!50] at (-8.5,2.5){};
\draw[dashed,line width=0.75pt,color=red] (-8.64,2) -- (-8.3,2);
\node[text width=5cm] at (-5.5,4) {MTD cluster head};
\node[text width=5cm] at (-5.5,3.5) {MTD cluster member};
\node[text width=5cm] at (-5.5,3) {M2M link};
\node[text width=5cm] at (-5.5,2.5) {Interfering MTD};
\node[text width=5cm] at (-5.5,2) {Interfering M2M link};

\end{tikzpicture}
\caption{Illustration example of clustering in an M2M network.\vspace{-0.5cm}}
\label{fig:SystemModel}
\end{figure}

We consider a reference cluster head, MTD $i$ located at $\boldsymbol{y}_i$, with an M2M link coverage area of radius $r_n$.  We model the area around MTD $i$ as a ball $b(\boldsymbol{y}_i,r_n)$, which forms a typical cluster $\mathcal{K}_{i,n}$ of $n$ MTDs, where $n=\lfloor\lambda_m\pi r_n^2\rfloor$.  Those MTDs within the ball are independently and uniformly distributed.  The maximum number of MTDs in a cluster is $N$, where $N=\lfloor\lambda\pi r_{N}^2\rfloor$, and $r_{N}$ is the maximum distance from the reference cluster head $i$, where MTDs can have a reliable communication with the cluster head.  A typical cluster, models a finite circular area, where MTDs within a cluster have highly correlated data.  The typical cluster is defined as follows:

\begin{definition}\label{def:cluster}
\textnormal{A \emph{typical cluster} $\mathcal{K}_{i,n}$ is defined as a cluster centered around reference cluster head $i$ located at $\boldsymbol{y}_i$ (the same point as the ball) and consists of $n$ MTDs.  The typical cluster $\mathcal{K}_{i,n}$ is a subset of the set of MTDs, $\mathcal{K}_{i,n}\subseteq\mathcal{M}$, that covers a circular area of radius $r_n$, with $n=\lfloor\lambda_m\pi r_n^2\rfloor$ MTDs.}
\end{definition}

From Definition~\ref{def:cluster}, we can see that the number of MTDs within the typical cluster $\mathcal{K}_{i,n}$ is a function of the cluster radius $r_n$ and the location of cluster head $\boldsymbol{y}_i$.  We assume that the MTDs form a collection of infinite disjoint clusters $\mathfrak{K}$.  A cluster $\mathcal{K}_{i,n}\in\mathfrak{K}$ is defined as a subset of $\mathcal{M}$, where $\bigcup_{\mathcal{K}_{i,n}\in\mathfrak{K}}\mathcal{K}_{i,n}=\mathcal{M}$.  Fig.~\ref{fig:SystemModel} represents an illustrative example of a ball and a typical cluster, in the M2M network.  Within the ball, we consider a typical cluster $\mathcal{K}_{i,n}\in\mathfrak{K}$ with cluster size $n$, and a reference cluster head located at the center of the ball and the typical cluster.  Note that, the index of the typical cluster directly relates to the cluster size, that is, the number of MTDs within the cluster including the reference cluster head.  As shown in Fig.~\ref{fig:SystemModel}, MTDs outside of the ball will cause inter-cluster interference to the reference cluster head $i$.  As illustrated in Fig.~\ref{fig:SystemModel}, we consider an example scenario, where cluster $\mathcal{K}_{i,n}$ has a radius $r_{n}$ and a maximum cluster radius of $r_N$.  In cluster $\mathcal{K}_{i,n}$ and cluster $\mathcal{K}_{j,n}$, MTDs $4$ and $6$ respectively, share the same resource block, and MTD $6$ from cluster $\mathcal{K}_{j,n}$ causes interference to cluster head $i$ in cluster $\mathcal{K}_{i,n}$.

We assume that, within any typical cluster $\mathcal{K}_{i,n}\in\mathfrak{K}$ during time slot $T$, each MTD $m$ must be able to send $H_m$ bits, as given by (\ref{entropy}), to the reference cluster head $i$.  Therefore, the transmission rate $R_{mi}$ of the M2M link between each cluster member $m\in\mathcal{K}_{i,n}\setminus\{i\}$ and cluster head $i$, must be greater than or equal to the threshold bit rate $\frac{H_m}{T}$, that is, $\forall m\in\mathcal{K}_{i,n}\setminus\{i\}:~R_{mi}\geq \frac{H_m}{T}$.  Based on the Gaussian random field model, the collection of data streams, $\boldsymbol{s}_i=[s_m]_{n \times 1}$, from all of the MTDs in cluster $\mathcal{K}_{i,n}$ will follow a multi-variate Gaussian distribution.  Here, MTDs within close proximity may gather correlated data~\cite{hsieh2015not}.  The covariance matrix of the data from the MTDs in each cluster will be $\boldsymbol{\Sigma}_{\mathcal{K}_{i,n}}=\left[\sigma_{ml}\right]_{n\times n}, \forall m,l \in \mathcal{K}_{i,n}$, where $\sigma_{ml}=\sqrt{\frac{\kappa}{\kappa+d_{ml}}}, \forall m,l \in \mathcal{K}_{i,n}$, $d_{ml}$ is the distance between MTD $m$ and $l$, and $\kappa$ is the correlation factor.  Thus, as the distance between MTDs decreases, the covariance will increase.  To model the joint information of data within each cluster, we use the notation of a joint entropy derived in~\cite{pattem2008impact}.  Since the number of MTDs considered in our system model is massive and the location of MTDs are random, deriving the exact distance between the cluster members and the cluster head can be challenging.  In order to simplify our model we assume the worst-case scenario of joint information in each cluster, that is, we assume that the distance between each cluster member and the cluster head is the maximum cluster radius, $r_n$.  Hence, the worst-case joint entropy $H_{\mathcal{K}_{i,n}}$ for each cluster will be given by~\cite{pattem2008impact}:
\begin{equation}\label{jointEntropy}
H_{\mathcal{K}_{i,n}}=H_m+H_m(n-1)\left(1-\frac{\alpha}{\frac{r_n}{c}+1}\right),
\end{equation}
where $\alpha=\frac{\log_2(e)}{\log_2(2\pi e)}$, $c$ is a correlation constant, and all MTDs have the same individual entropy $H_m$.

The transmission power per resource block, $q_{mi}$, for MTD $m$ in cluster $\mathcal{K}_{i,n}$ that is used to send $H_m$ bits over an M2M link to the reference cluster head $i$, is given by:
\begin{equation}\label{machine_TxP}
q_{mi}(t) = \left(\frac{I_{-i}(t)}{g_{mi}(t)}\right)\left( 2^{\frac{H_m}{TB}}-1\right),
\end{equation}
where $I_{-i}(t)$ is the inter-cluster interference from MTDs in clusters outside $\mathcal{K}_{i,n}$, if those MTDs use the same resource blocks that MTDs in cluster $\mathcal{K}_{i,n}$ use to send their data to cluster head $i$, thus $\mathcal{K}_{-i}=\{\mathcal{K}_j|\forall\mathcal{K}_j\in\mathfrak{K},\mathcal{K}_j\neq \mathcal{K}_i\}$.  The channel gain between MTD $m$ and MTD $i$ is given by $g_{mi}(t)=\xi(t)A_t\left(\frac{4\pi d_{mi}}{\mu}\right)^{-\nu}$, where $\nu$ is the path loss exponent, $d_{mi}$ is the distance between MTD $m$ and $i$, $\mu$ is the wavelength of an electromagnetic wave, $\xi(t)$ is the time-varying fading channel attenuation, and $A_t$ is the antenna gain of the transmitter and receiver. For simplicity, we assume that the thermal noise is negligible compared to the interference and hence is ignored.  For a quick reference, the notation used in this paper is summarized in Table~\ref{tab:notation}.

\begin{table}[t!]
\vspace{-0.4cm}
\begin{center}
\caption{Summary of the notations used throughout this work.\vspace{-0.3cm}}
\begin{tabular}{|c|p{14cm}|}
\hline\hline
\textbf{Notation}   & \textbf{Description} \\ \hline
$\phi_M$            & PPP modeling the locations of MTDs. \\ \hline
$\lambda_m$         & Density of $\phi_M$. \\ \hline
$\phi_{I_n}$        & PPP modeling the locations of interfering MTDs.\\ \hline
$\lambda_{In}$      & Density of $\phi_{I_n}$. \\ \hline
$g_{mi}(t)$         & Channel gain power between MTD $m$ and MTD $i$ at time slot $t$. \\ \hline
$\nu$               & Path loss exponent for all M2M and cellular links, where $\nu>2$.  \\ \hline
$\mathcal{K}_{i,n}$ & Typical cluster with radius $r_n$ and MTD $i$ as the cluster head; subset of MTDs $\mathcal{K}_{i,n}\subset\mathcal{M}$. \\ \hline
$n$     			& Size of a typical cluster $\mathcal{K}_{i,n}$, where $n=\lfloor\lambda_m\pi r_n^2\rfloor$. \\ \hline
$N$                 & Maximum number of MTDs per cluster.\\ \hline
$r_{N}$			    & Maximum distance from the reference cluster head $i$, i.e., the maximum cluster radius. \\ \hline
$p_m$; $q_m$        & Transmit power via the: cellular link, $p_m\in[0,P_{\max}]$, and the M2M link, $q_m\in[0,Q_{\max}]$.\\ \hline
$H_m$				& Individual joint entropy of MTD $m$.\\ \hline
$H_{\mathcal{K}_{i,n}}$ & Joint entropy of cluster $\mathcal{K}_{i,n}$.\\
\hline\hline
\end{tabular}
\label{tab:notation}
\end{center}
\end{table}

We derive an expression for M2M link distance within the typical cluster $\mathcal{K}_{i,n}$ and the reference cluster head $i$, by substituting the channel gain expression, $g_{mi}(t)$ into (\ref{machine_TxP}), and rearranging for $d_{mi}$, as follows:
\begin{align}\label{M2M-tx}
q_{mi}(t) &= \left(\frac{I_{-i}(t)}{\xi(t)A_t\left(\frac{4\pi d_{mi}}{\mu}\right)^{-\nu}}\right)\left( 2^{\frac{H_m}{TB}}-1\right), \nonumber \\
\therefore d_{mi} &= \frac{\mu \left(q_{mi}(t)\xi(t)A_t\right)^{\frac{1}{\nu}} }{4\pi \left(\left(2^{\frac{H_m}{TB}}-1\right)I_{-i}(t)\right)^{\frac{1}{\nu}}}.
\end{align}
Given (\ref{M2M-tx}) and assuming maximum M2M link transmission power is used $q_{mi}(t)= Q_{\max}$ for all MTDs within cluster $\mathcal{K}_{i,n}$, the worst-case scenario for maximum M2M link distance can be derived, which will result in maximum joint entropy.  Hence, the maximum distance to successfully send $H_m$ bits over an M2M link between MTD $m$ and cluster head $i$, is given by:
\begin{align}\label{Max-M2M-Dist}
D_{mi}=\frac{\mu \left(Q_{\max}\xi(t)A_t\right)^{\frac{1}{\nu}} }{4\pi \left(\left(2^{\frac{H_m}{TB}}-1\right)I_{-i}(t)\right)^{\frac{1}{\nu}}}.
\end{align}

Thus, the maximum radius for reference cluster head $i$ is $r_N=D_{mi}$.  The total transmission power of a typical cluster $\mathcal{K}_{i,n}$, can be defined as the summation of transmission powers over all M2M links and the cellular link within cluster $\mathcal{K}_{i,n}$.  Following (\ref{cellular_TxP}) and (\ref{machine_TxP}), the total transmission power $P_{\mathcal{K}_{i,n}}(t)$ of cluster $\mathcal{K}_{i,n}$ is given by:
\begin{align}\label{Power-K}
P_{\mathcal{K}_{i,n}}(t)= \sum_{m\in\mathcal{K}_{i,n}\setminus\{i\}}\left(\frac{I_{-i}(t)}{g_{mi}(t)}\right)\left(2^{\frac{H_m}{TB}}-1\right)+
 \left(\frac{BN_0}{g_{i}(t)}\right)\left(2^{\frac{H_{\mathcal{K}_{i,n}}}{TB}}-1\right).
\end{align}

As the correlation between MTD data within the cluster increases, this causes the total transmission power of the typical cluster $\mathcal{K}_{i,n}$ to also increase, which is due to the increasing number of MTDs within the typical cluster.  From (\ref{Power-K}), with decreasing joint entropy, the cellular link transmit power will also decrease. Thus, there are many benefits, such as, reduced transmit power per MTD and increasing redundant bits, when clustering MTDs based on correlation in an M2M network.  However, finding an optimal and centralized cluster formation for an infinite number of MTDs in a densely deployed M2M network is challenging.  Centralized clustering methods can potentially increase the complexity and signaling overhead due to the density of M2M networks. The signaling overhead can further increase due to the dynamics of the MTD network that can stem from various factors, such as the joining of new MTDs, MTD loss of battery, or rapid fluctuations in the sensing environment.  However, any fully distributed correlation-aware clustering of MTDs also depends on power allocation, location of MTDs, and the correlation between MTD data.  Thus, determining how MTDs form clusters in a fully distributed manner, when the number of MTDs is large and their location is random, is challenging.

Next, we propose a fully distributed correlation-aware MTD clustering framework based on \emph{evolutionary game theory}~\cite{khan2011evolutionary,luo2013evolutionary,platkowski2016evolutionary} for an infinite number of MTDs.  Accurate modeling of the M2M network topology becomes a key step towards meaningful performance analysis of correlation-aware clustering.  To perform such modeling, prior to formulating the game, we first use stochastic geometry tools~\cite{haenggi2012stochastic,chiu2013stochastic} to characterize metrics, such as the distribution of distances between MTDs and inter-cluster interference.  Then, based on these system parameters, we propose an evolutionary game model that can effectively capture the dynamics pertaining to an infinite number of possible MTD cluster formations.  At the convergence point, an evolutionary game is robust to potential deviations of MTDs that can result from factors such as, MTDs leaving/entering the network, or rapid fluctuations in the sensing environment.  The proposed model based on stochastic geometry and evolutionary game theory, can help network designers predict how design parameters, such as the density of MTDs, transmission power, resource block bandwidth, and duration time slot, will affect the evolutionary stable clusters of MTDs.

\section{Correlation-Aware Evolutionary Game in an Infinite M2M network} \label{sec:EGgame}
In this section, we propose an evolutionary game~\cite{khan2011evolutionary,luo2013evolutionary,platkowski2016evolutionary,jiang2014evolutionary,tembine2010evolutionary,young2014handbook,sartakhti2017mmp} for clustering an infinite number of MTDs, in order to reduce transmission power for each MTD as well as the number of redundant bits sent to the BS.  This approach is fully distributed, as it allows an infinite number of MTDs to self-organize into clusters based on increasing data correlation and reducing transmission power.  In order to define a tractable utility expression for this evolutionary game, we first use stochastic geometry, we characterize the distribution of distances from the reference cluster head to all cluster members.  Then, we determine the distribution of the inter-cluster interference to all cluster members, excluding the reference cluster head.  Since the number of MTDs within the network is massive, we cannot easily and accurately obtain all MTD locations in the network.  By characterizing the distance distributions between MTDs, we can precisely analyze the data correlation within each cluster.  In fact, we observe that, as the distance between MTDs decreases, the covariance between MTD data will increase, which in turn, causes the joint entropy of the cluster to decrease.  Using this analysis, we can then derive the utility function for our game and use it to analyze the outcome of the MTDs' interactions.

To model the distance distribution of each cluster member in any typical cluster $\mathcal{K}_{i,n}$ to the reference cluster head $i$, we consider the $k$-closest MTD approach, where MTD $m\in\mathcal{K}_{i,n}\setminus\{i\}$ is the $k^{th}$ closest MTD to cluster head $i$.  Therefore, we need to ``order" the distances from cluster head $i$ to all MTDs in cluster $\mathcal{K}_{i,n}$, in increasing order.  We define an \emph{ordered set} of distances $\mathcal{D}_i=\{d_{(k)i}\}_{k=1:n-1}$, by sorting $d_{mi}$ in ascending order, such that $d_{(1)i}\leq d_{(2)i}\leq \ldots\leq d_{(k)i}\leq\ldots\leq d_{(n-1)i}$.  Since the random variables of the sequence $\mathcal{D}_i$ are independent and identically distributed (i.i.d.), the probability distribution function (PDF) of the distance from the $k^{th}$ closest MTD to cluster head $i$, will be given by~\cite{srinivasa2010distance}:
\begin{align}\label{K_Distance}
f_{\mathcal{D}_i}^{(k)}(d)=\frac{(n-1)!}{(k-1)! (n-1-k)!}{F_{\mathcal{D}_i}(d)}^{k-1}f_{\mathcal{D}_i}(d)\big(1-{F_{\mathcal{D}_i}(d)}\big)^{n-1-k},
\end{align}
where $d=d_{(k)i}$, $f_{\mathcal{D}_i}(d_{(k)i})=\frac{2d_{(k)i}}{r_n^2}$ for $0 \leq d_{(k),i} \leq r_n$, and $F_{\mathcal{D}_i}(d_{(k)i})=\frac{d_{(k)i}^2}{r_n^2}$ for $0 \leq d_{(k),i} \leq r_n$.  Next, we characterize the distribution of the inter-cluster interference to the reference cluster head.

The worst-case inter-cluster interference for all M2M links in cluster $\mathcal{K}_{i,n}$ is considered.  We assume that the M2M links in different clusters $\mathcal{K}_{-i}=\{\mathcal{K}_j |\forall\mathcal{K}_j\in\mathfrak{K},\mathcal{K}_j\neq\mathcal{K}_i\}$ simultaneously reuse the same resource blocks.  Therefore, the distance between an interfering MTD from cluster $\mathcal{K}_{-i}$ to cluster head $i\in\mathcal{K}_{i,n}$ is greater than the cluster radius, $r_n$.  In the worst-case, one M2M link in each cluster uses the same resource block, thus the number of interfering MTDs is equal to the number of other clusters in $\mathcal{K}_{-i}\in\mathfrak{K}$, which is less than the total number of MTDs.  If MTDs form cluster with radius $r_n$, the location of interfering MTDs follows PPP, $\Phi_{I_n}$, with density $\lambda_{I_n}=\frac{1}{\pi r_n^2}$.  This is due to the fact that there is one randomly deployed interferer MTD in each cluster~\cite{haenggi2012stochastic}.  Based on the PPP $\phi_{I_n}$, the mean interference for a typical cluster head, is given by~\cite{haenggi2012stochastic}:
\begin{align}
\mathbb{E}(I_n) = \mathbf{E} \left( \sum_{\boldsymbol{y} \in \Phi_{I_n}} q A_t \xi(t) \left(\frac{4\pi}{\mu}\right)^{-\nu} \|\boldsymbol{y}\| ^{-\nu} \right)= q A_t \xi(t)\left(\frac{4\pi}{\mu}\right)^{-\nu}  \int_{\mathds{R}^2 \backslash  b(\boldsymbol{y}_i,r_n)}  \lambda_{I_n} \|\boldsymbol{x}\| ^{-\nu}   dx, \nonumber
\end{align}
\begin{equation}\label{Interferense}
\therefore \mathbb{E}(I_n) = q A_t\xi(t)\left(\frac{4\pi}{\mu}\right)^{-\nu} \left(\frac{1}{\pi r_n^2}\right) \int_{r_n}^\infty \int_0^{2\pi} r ^{-\nu}   r d\theta dr,
\end{equation}
where $r_n$ is the radius of a typical cluster.  Lets assume $2<\nu$ and $2 \neq \nu$, then the mean interference for a typical cluster head with radius $r_n$, is given by~\cite{haenggi2012stochastic}:
\begin{equation}
\mathbb{E}(I_n)= 2q A_t \xi(t)\left(\frac{4\pi}{\mu}\right)^{-\nu} \left(\frac{1}{r_n^\nu(\nu-2)}\right).
\label{Interferense_simple}
\end{equation}

The mean inter-cluster interference in (\ref{Interferense_simple}), is a function of the typical cluster radius and the path loss exponent.  As the network density increases, this causes the inter-cluster interference to also increase due to the increasing number of MTDs in the network.  Additionally, as the radius of the typical cluster increases, this causes the inter-cluster interference to reduce, due to reducing the number of interfering MTDs outside the cluster.

Now, having derived the needed metrics, next, we define the notion of a cluster type that is needed to formally define our proposed evolutionary game.  Then, we use the proposed game to study the distributed clustering problem for an infinite number of MTDs, while taking into account MTD data correlation and transmission power.
\begin{definition}~\label{def:type}
\textnormal{A \emph{cluster type} $j\in\mathcal{S}$ represents the number of MTDs within a cluster, that is, a cluster of size $j$.}
\end{definition}

Thus, the proposed evolutionary game can be formally defined as follows:
\begin{definition}
\textnormal{An \emph{evolutionary game} is defined by, $G_{E}=(\mathcal{M},\mathcal{S},\boldsymbol{x},u)$, where $\mathcal{M}$ is an infinite set of MTDs (population), $\mathcal{S}$ is the finite set of cluster types (pure strategy set), $\boldsymbol{x}$ is the population state, and $u$ is the utility function of a cluster.}
\end{definition}

In our proposed evolutionary game, the finite set of cluster types for a cluster is, $\mathcal{S}=\{1,2,\ldots,N\}$, which denotes the set of potential cluster sizes.  Note that, $N= \lfloor\lambda_m\pi r_{N}^2\rfloor$ is the maximum potential cluster size for any cluster head within the ball, with $r_{N}$ being the maximum M2M range which is equal to $D_{mi}$ in (\ref{Max-M2M-Dist}).

The population state vector $\boldsymbol{x}\in\mathds{R}^{N}$, where $\sum_{j\in\mathcal{S}}x_j=1$, captures the percentage of MTDs forming potential cluster sizes.  Thus, each element $x_j$ of $\boldsymbol{x}$ represents the average percentage of MTDs forming a cluster size $j$, according to Definition~\ref{def:type}.  The utility achieved by a given cluster $\mathcal{K}_{i,j}$ at time slot $t$ is denoted as $u_{\mathcal{K}_{i,j}}(t)$.  Note that, within the ball $b(\boldsymbol{y}_i,r_j)$, we consider one cluster whose cluster head is located at $\boldsymbol{y}_i$ in the network.

We need to derive a utility function that captures the average transmission power per MTD per cluster $\mathcal{K}_{i,j}$.  The defined utility function will include two main terms.  In the first term, we model the average transmission power of M2M links across all MTDs $m\in\mathcal{K}_{i,j}\setminus\{i\}$ to the reference cluster head $i$, as defined in (\ref{machine_TxP}).  In the second term of the utility function, we capture the cellular link transmit power of cluster head $i$ to the BS, as defined in (\ref{cellular_TxP}).  Since we are considering an infinite number of MTDs, the utility function must also take into account the MTD distance distributions, as defined in (\ref{K_Distance}), and inter-cluster interference, as defined in (\ref{Interferense_simple}), in order to accurately model MTD average transmission power.  Thus, the closed-form expression of the proposed utility function at time slot $t$, $u_{\mathcal{K}_{i,j}}(t)$, can be derived as follows.

\begin{theorem}\label{Theor_utility}
\textnormal{The closed-form expression of the utility function $u_{\mathcal{K}_{i,j}}(t)$ when typical cluster $\mathcal{K}_{i,j}$ chooses type $j$, is expressed as follows:
\begin{align}\label{Utility-closed}
u_{\mathcal{K}_{i,j}}(t)=&-\frac{1}{j}\left(\frac{1}{A_t\xi(t)}\right)\left(\frac{4\pi r_j}{\mu}\right)^{\nu}\left(\frac{2\left(j-1\right)\mathbb{E}(I_j)}{2+\nu}\right) \left(2^{\frac{H_m}{TB}}-1\right) \nonumber \\
&-\frac{1}{j}\left(\frac{BN_0}{g_{i}}\right)\left(2^{\frac{H_m+H_m(j-1)\left(1-\frac{\alpha}{\frac{r_j}{c}+1}\right)}{TB}}-1\right),
\end{align}
where $j= \lfloor\lambda_m\pi r_j^2\rfloor$.}
\end{theorem}

\begin{proof}
\textnormal{See Appendix~\ref{app:A}.}
\end{proof}

The proposed utility function (\ref{Utility-closed}), is the average transmission power per MTD per cluster $\mathcal{K}_{i,j}$ selecting cluster size $j$, during time slot $t$.  Depending on the cluster type of the typical cluster, this will affect the size of the cluster and the average transmission power per MTD per cluster.  Thus, as the cluster size increases (that is, increasing the cluster type), this causes average transmission power and signaling overhead to increase, due to the increasing number of MTDs within the cluster sharing their data over the M2M link.

\subsection{Dynamics of Cluster Formation}
Evolutionary game theory uses biologically-inspired dynamics to model how individual MTDs form different types of clusters (i.e., cluster of different sizes) over time~\cite{tembine2010evolutionary,young2014handbook}.  In our MTD game, we assume that the percentage of MTDs that select cluster type $j\in\mathcal{S}$ is $x_j$, where $x_j\in[0,1]$ is an element of the population state vector $\boldsymbol{x}$.

To update the percentage of MTDs selecting cluster type $j$, we adopt properties from continuous-time replicator dynamics (see~\cite{tembine2010evolutionary,young2014handbook}, and~\cite{niyato2009dynamics}). Replicator dynamics are used to model the rate of cluster $\mathcal{K}_{i,j}$ selecting strategy type $j$ from the set of cluster types $\mathcal{S}$.  Over time, the MTDs in a cluster $\mathcal{K}_{i,j}$ will update their preference in $\boldsymbol{x}$ and will become more certain about what cluster type they would prefer to form~\cite{niyato2009dynamics,tembine2010evolutionary}.  In continuous-time replicator dynamics, the rate of MTDs selecting cluster type $j$ is proportional to the difference between the fitness of cluster type $j$ and the average expected fitness of the population~\cite{sartakhti2017mmp,young2014handbook}.  Note that, fitness of a type is defined as the average payoff of that type, and is a function of the population state $\boldsymbol{x}$.  We use replicator dynamics to model the evolution of MTD cluster size preferences, based on high data correlation and reduced transmission power~\cite{yan2017evolutionary,young2014handbook,tembine2010evolutionary}.  The evolution of cluster type $j$, is given by:
\begin{equation}\label{xdot}
\dot{x}_j(t) = x_j(t)(\bar{u}_{\mathcal{K}_{i,j}}(\boldsymbol{x},t)-U_{\mathcal{K}_i}(\boldsymbol{x},t)),
\end{equation}
where $\bar{u}_{\mathcal{K}_{i,j}}(\boldsymbol{x},t)$ is the fitness of cluster $\mathcal{K}_{i,j}$ with size $j$, and $U_{\mathcal{K}_i}(\boldsymbol{x},t)$ is the average expected fitness of the population.  Thus, the fitness of cluster type $j$, is defined as:
\begin{equation}\label{fitness}
\bar{u}_{\mathcal{K}_{i,j}}(\boldsymbol{x},t) = \sum_{j^\prime\in\mathcal{S}}u_{\mathcal{K}_{i,jj^\prime}}(t)x_{j^\prime}(t),
\end{equation}
where $u_{\mathcal{K}_{i,jj^\prime}}(t)=u_{\mathcal{K}_{i,j}}$ if $u_{\mathcal{K}_{i,j}} \geq u_{\mathcal{K}_{i,j^\prime}}$, or $u_{\mathcal{K}_{i,jj^\prime}}(t)=u_{\mathcal{K}_{i,j^\prime}}$ if $u_{\mathcal{K}_{i,j^\prime}} > u_{\mathcal{K}_{i,j}}$; $x_{j^\prime}$ is the population state of cluster type $j^\prime\in\mathcal{S}$ in the evolutionary game; and $u_{\mathcal{K}_{i,j}}$ is given by (\ref{Utility-closed}).  Furthermore, the average expected fitness of the population, $U_{\mathcal{K}_i}(\boldsymbol{x},t)$, is given by:
\begin{equation}\label{Eutility}
U_{\mathcal{K}_i}(\boldsymbol{x},t) = \sum_{j\in\mathcal{S}}\bar{u}_{\mathcal{K}_{i,j}}(\boldsymbol{x},t)x_j(t).
\end{equation}

The evolution of cluster type $j$ for the MTD population in (\ref{xdot}) can be either greater than, less than, or equal to zero.  If $\dot{x}_j$ is greater than $0$, this implies that the fitness of cluster type $j$ is greater than the average expected fitness of the population, and thus the percentage of the population selecting cluster type $j$ is increasing.  If $\dot{x}_j$ is less than $0$, this indicates that the percentage of the population selecting cluster type $j$ is decreasing, that is, cluster type $j$ is growing extinct.  When $\dot{x}_j$ is equal to $0$ then, we have a stationary point for the percentage of MTDs selecting cluster type $j$ (an evolutionary equilibrium for cluster type $j$).  Thus, no MTD has incentive to change their cluster type from $j$, as the fitness of cluster type $j$ is equal to the average expect fitness of the population.

To solve the proposed evolutionary game, we propose a fully distributed algorithm to find a stable cluster formation for correlated-aware clustering in M2M communications.  The proposed distributed algorithm, shown in Algorithm~\ref{algorithm-E}, enables an infinite number of MTDs to autonomously update their type to form the best cluster.  Algorithm~\ref{algorithm-E} outputs the population state, $\boldsymbol{x}^\circ$, which includes the percentage of MTDs in the population selecting a cluster type from the strategy set.

\subsection{Evolutionary Game Stability Analysis}\label{Sec:Stable}
To solve the evolutionary game, we consider the concept of an ESS.  The ESS is robust to a small portion $\epsilon_x\in(0,1),\forall~\epsilon\in(0,\epsilon_x)$, of MTDs changing their cluster type.  This change in cluster type could be due to the joining of new MTDs, rapid fluctuations in the sensing environment, or MTD loss of battery.  The ESS is defined as follows:

\begin{algorithm}[!t]
\caption{Evolutionary Game for Correlation-Aware Clustering in Massive M2M Network} \label{algorithm-E}
\begin{algorithmic}[1]
\State \textbf{Input:} Network density: $\lambda_m$; Set of MTDs: $\mathcal{M}$; Maximum radius: $r_N$;
\State Place a ball $b(\boldsymbol{y}_i,r_n)$ randomly in the network, centered at $\boldsymbol{y}_i, i\in\mathcal{M}$ with radius $r_n$, where MTD $i$ is the cluster head.
\State Find the maximum number of MTDs, $N$, within the ball, given the network density $\lambda_m$ and radius of the ball $r_N$, $N=\lfloor\lambda_m\pi r_N^2 \rfloor$;
\State Define the set of cluster types for a cluster as $\mathcal{S}=\{1,2,\ldots,N\}$, and set the initial population state for all cluster types as $\boldsymbol{x}(t=0)$;
\Repeat
\For{cluster type $j\in 1\ldots N$}
\parState{%
Find the inter-cluster interference, $\mathbb{E}(I_j)$, given the network density $\lambda_m$ and radius of the ball $r_j$;}
\parState{%
For a cluster $\mathcal{K}_{i,j}$ that chooses cluster type $j$, calculate the utility function, $u_{\mathcal{K}_{i,j}}(t)$, as in (\ref{Utility-closed});}
\parState{%
For cluster type $j$ determine the evolution of MTD preference, $\dot{x}_{j}(t)$, using (\ref{xdot}), (\ref{fitness}), and (\ref{Eutility});}
\State Update population state for cluster type $j$: $x_{j}(t+1)=x_{j}(t)+\dot{x}_{j}(t)$;
\EndFor
\State $t=t+1$;
\Until{~Convergence to ESS and $\dot{x}_{j}(t)=0~\forall j\in\mathcal{S}$}
\State \textbf{Output:} Population state, $\boldsymbol{x}^\circ=\{x_j|\forall j\in\mathcal{S}\}$, that represents the formed MTD clusters.
\end{algorithmic}
\end{algorithm}

\begin{definition}\label{ESS_def}
\textnormal{The population state $\boldsymbol{x}^*\in\mathds{R}^{N}$ is an \emph{ESS}, if there exists a portion of MTDs $\epsilon^*_j>0$ for each cluster size $j\in\mathcal{S}$, such that for all $0<\epsilon<\epsilon^*_j$ and for all $j\in\mathcal{S}$:
\begin{align}\label{ESS}
U_{\mathcal{K}_i}(x_j^*(t),(1-\epsilon)\boldsymbol{x}_{-j}^*(t)+\epsilon\boldsymbol{x}^\prime_{-j}(t))> U_{\mathcal{K}_i}(x_j^\prime(t),(1-\epsilon)\boldsymbol{x}_{-j}^*(t)+\epsilon\boldsymbol{x}_{-j}^\prime(t))
\end{align}
where $\boldsymbol{x}^\prime(t)\in\mathds{R}^{N}$ is any population state which is different from $\boldsymbol{x}^*(t)$, i.e., $\boldsymbol{x}^\prime(t)\neq\boldsymbol{x}^*(t)$ and $x^\prime_j(t)$ is an element of $\boldsymbol{x}^\prime(t)$.  In particular, $\epsilon\boldsymbol{x}^\prime(t)$ represents portion $\epsilon$ of MTDs from the population, that will choose a strategy from the population state $\boldsymbol{x}^\prime(t)$ instead of $\boldsymbol{x}^*(t)$.}
\end{definition}

The replicator dynamics in (\ref{xdot}) capture the dynamics of distributed MTD clustering in our proposed algorithm.  Algorithm~\ref{algorithm-E} converges to percentages of various cluster sizes within the MTD population, and if so to which percentages~\cite{shoham2008multiagent}.  Next, we prove that the proposed Algorithm~\ref{algorithm-E} converges to a population state $\boldsymbol{x}^*$, and we find the maximum portion of MTDs, $\epsilon$, that may deviate from an ESS, based on Definition~\ref{ESS_def}.

\begin{theorem}\label{Theor_band_epsilon}
\textnormal{The proposed Algorithm~\ref{algorithm-E} converges to a population state, $\boldsymbol{x}^*$, which is an ESS for the proposed evolutionary correlation-aware clustering game in $G_{E}$. At the convergence, the maximum portion of MTDs that can deviate from ESS, $\mathcal{\epsilon}^*$, is given as follows:}.
\begin{equation}\label{Band_epsilon}
  \epsilon^*=\min_{j\in \mathcal{S}}
  \frac{1}{
2\bar{u}_{\mathcal{K}_{i,j}}(\boldsymbol{x}^*,t)-u_{\mathcal{K}_{i,j}}(\boldsymbol{x}^*,t)}.
\end{equation}
\end{theorem}

\begin{proof}
\textnormal{See Appendix~\ref{app:B}.}
\end{proof}

\vspace{-0.3cm}
\section{Simulation Results}\label{Sec:Sim}
For our simulations, we consider a single BS located at the center of a circular area with a 2~km radius.  We consider a Poission-based distribution for distributing MTDs around a randomly placed MTD within the network.  We focus on a small section of the circular area around the reference cluster head, that has a 500~m radius.  The BS allocates orthogonal resource blocks to cluster heads that they can use to send the data of the clusters to the BS.  Within each cluster, each MTD is also allocated an orthogonal resource block, to send its data to the cluster head.  However, M2M links in different clusters can simultaneously reuse the same resource blocks, resulting in inter-cluster interference.  Each resource block has a fixed bandwidth of $B=180$~kHz, and the maximum transmission power over the cellular link is $P_{\max}=35$~dBm and the M2M link is $Q_{\max}=20$~dBm.  The duration of each time slot $t$ is fixed to $T=1$~ms. We consider a carrier frequency of $2$~GHz.  Furthermore, we assume MTDs are static, where the cellular and M2M links have a path loss exponent of $2.5$.  The cellular and M2M links have a transmit and receive antenna gain of $9.54$~dB.  The noise power spectral density over a cellular link is $-176$~dBm/Hz.  The data source of each MTD $m$ has $\mu_m=0$ and $\sigma_m=10$.  The quantization step for each MTD is set to $\Delta=\frac{1}{256}$~\cite{hsieh2015not}.

The evolutionary game is evaluated using Monte Carlo simulations with varying MTD density $\lambda_m$, path loss exponent $\nu$, and correlation constant $c$.  We compare our proposed evolutionary game at convergence to two benchmarks, which are: (i) a pure cluster type, where the population of MTDs prefer to select cluster type $N$ with preference $1$; and (ii) uniform cluster type, where an equal percentage of the MTD population are uniformly distributed across all cluster types, that is, each cluster type has preference $\frac{1}{N}$.

\begin{figure*}[!t]
 \centering
 \begin{multicols}{2}
 \includegraphics[width=\linewidth]{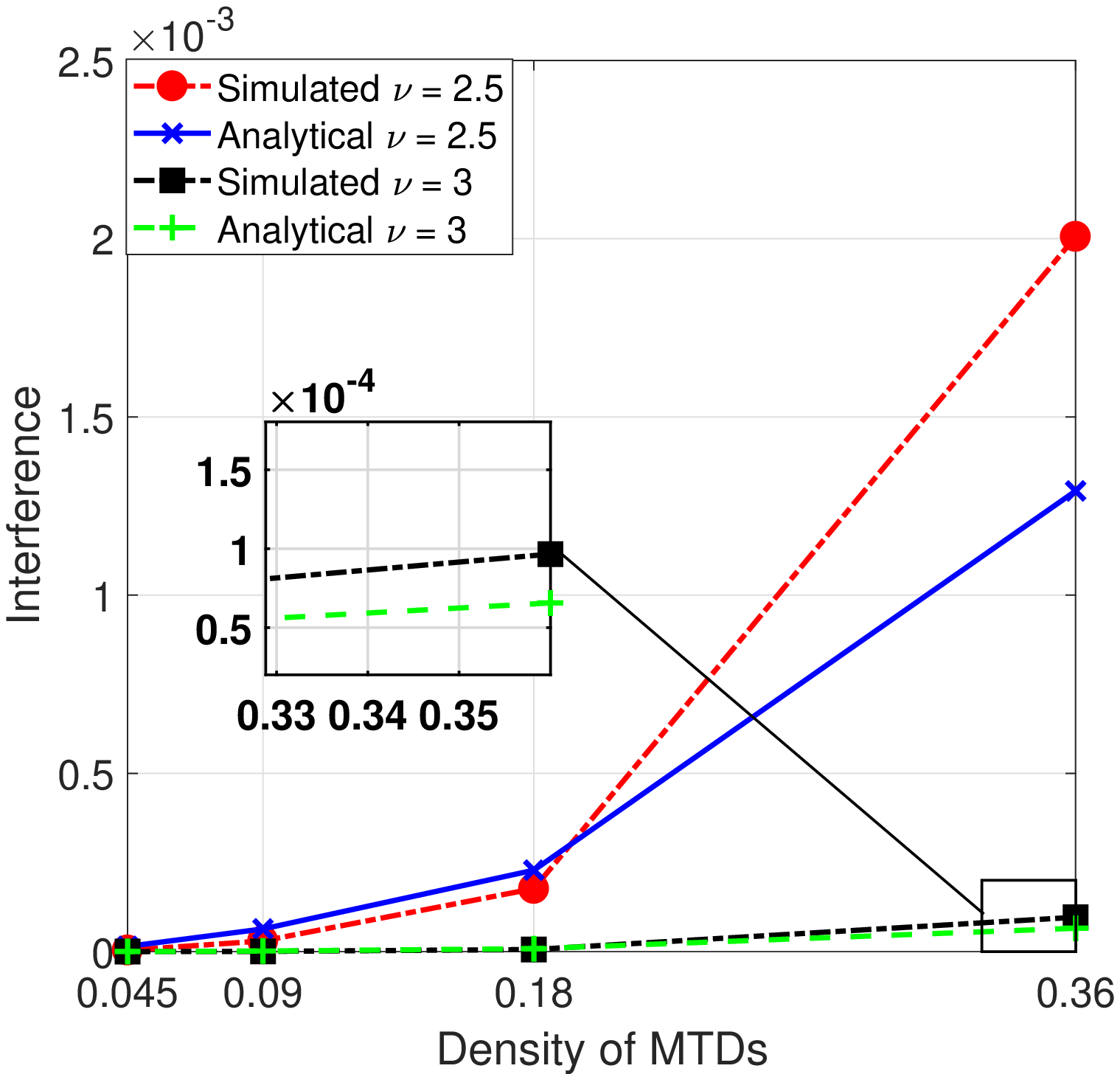}\par\vspace{-0.3cm}\caption{Simulated and analytical interference as a function of MTD density, $\lambda_m$, where $c=6$.}\label{fig:R1}
 \includegraphics[width=\linewidth]{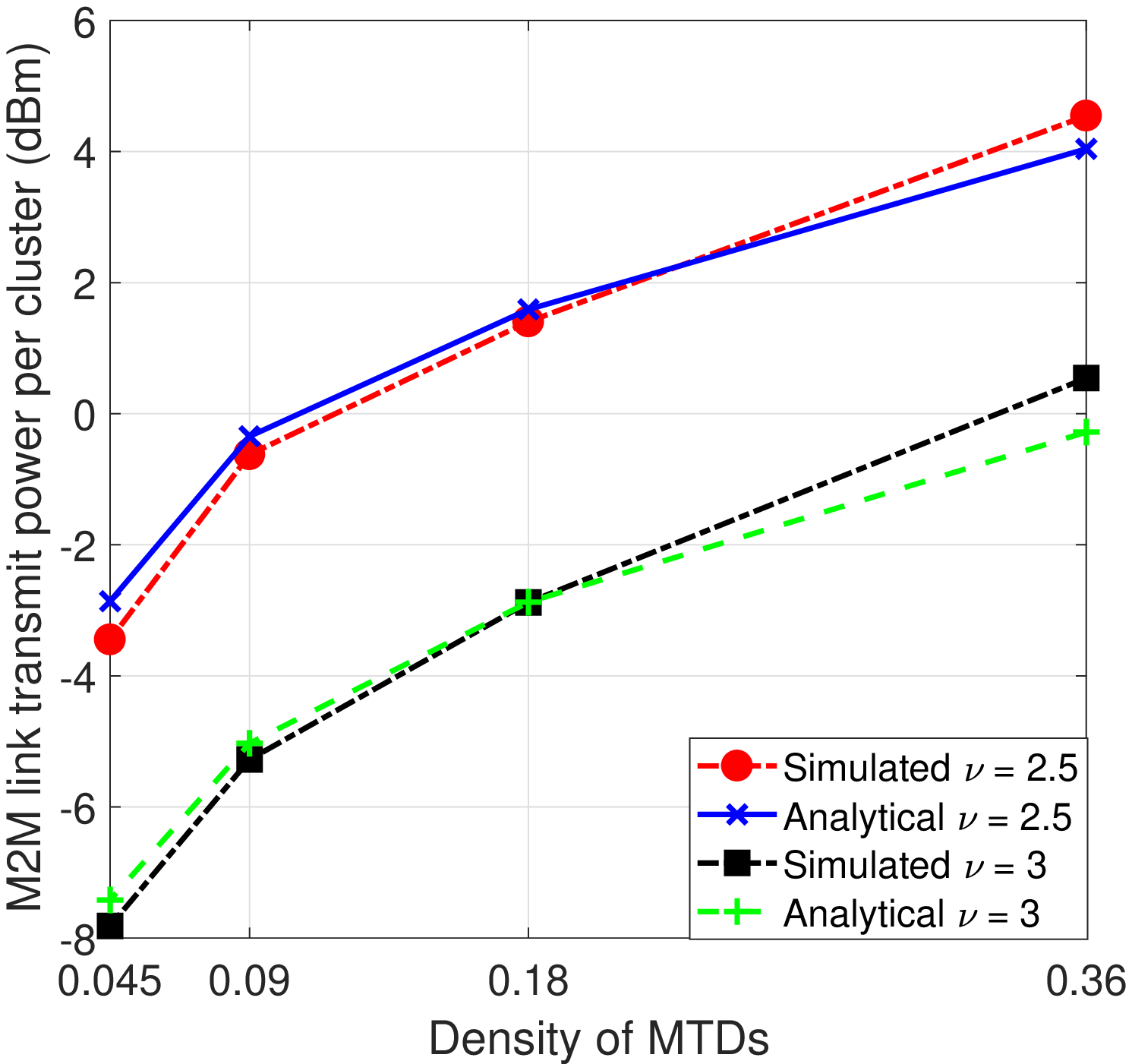}\par\vspace{-0.3cm}\caption{Simulated and analytical M2M link transmit power as a function of MTD density, $\lambda_m$, where $c=6$.}\label{fig:R1_1}
 \end{multicols}
 \vspace{-1cm}
\end{figure*}

Fig.~\ref{fig:R1} shows the inter-cluster interference for different network densities and path loss exponents.  From this figure, we can see that the analytical results derived using~(\ref{Interferense_simple}) closely match the corresponding simulation results.  Also, in Fig.~\ref{fig:R1}, we observe that, as the network density increases, this causes the inter-cluster interference to increase, due to the ball radius, $r_{N}$, decreasing as shown in Fig.~\ref{fig:R12}.  Furthermore, by increasing the path loss exponent from $\nu=2.5$ to $\nu=3$ the inter-cluster interference will be reduced due to the higher propagation losses.

Fig.~\ref{fig:R1_1} shows the M2M link transmit power per cluster for different network densities and path loss exponents.  From the figure, we can see that the analytical results derived using~(\ref{Utility-closed}) closely match the corresponding simulation results.  Also, in Fig.~\ref{fig:R1_1}, we observe that, as the network density increases, this causes the M2M link transmit power per cluster to increase, due to the increasing interference as shown in Fig.~\ref{fig:R1}.  Furthermore, by increasing the path loss exponent from $\nu=2.5$ to $\nu=3$ the M2M link transmit power per cluster will also be reduced due to the higher propagation losses.

In Fig.~\ref{fig:R3}, we show the probability density function (PDF) of MTD preference for forming certain cluster sizes within the ball, under different path loss exponents, $\nu$.  We observe that, as the path loss exponent increases, the maximum preference of MTDs forming a particular cluster size is the same.  As the path loss exponent increases, this causes bad channel modeling over the cellular and M2M links, and thus, causes MTDs to prefer to form smaller clusters.  We can see that, in our proposed evolutionary algorithm, MTDs prefer to form a cluster of size $5$ for $8.8$\% of the time for $\nu=2.5$, whereas MTDs prefer to form a cluster of size $5$ for $6.9$\% of the time, for $\nu=3$.  Overall, as the path loss exponent increases, MTDs prefer to form larger clusters.

\begin{figure}
  \centering
  \includegraphics[width=0.5\linewidth]{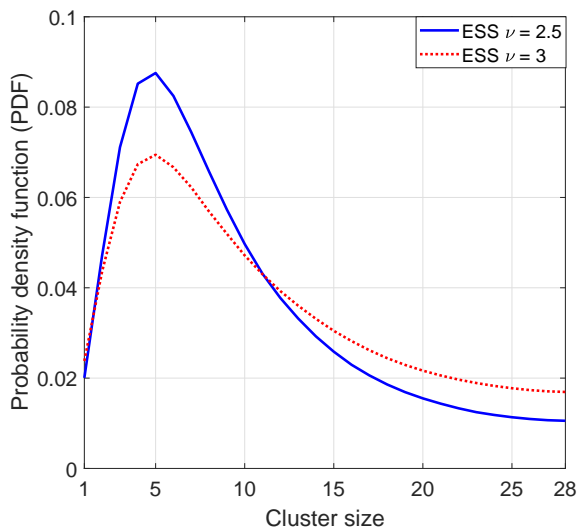}\par\vspace{-0.3cm}\caption{Percentage of MTDs forming a particular cluster size as a function of path loss exponents, $\nu$, where density of MTDs is $\lambda_m=0.09$ and $c=6$.\vspace{-0.5cm}}\label{fig:R3}
\end{figure}

Fig.~\ref{fig:R4} shows the probability density function (PDF) of MTD preference for forming certain cluster sizes within the ball, under different network densities, $\lambda_m$, and correlation constants, $c$.  We observe that, as the network density and correlation constant increase, the MTDs' preferences for forming particular cluster sizes does not change significantly.  This is due to the fact that the maximum radius of the ball, defined in (\ref{Max-M2M-Dist}), is a function of MTD density.  As illustrated in Fig.~\ref{fig:R12}, the radius of the ball changes according to the density, thus maintaining the number of MTDs within the ball consistent.  Furthermore, we observe that, in the proposed evolutionary algorithm, MTDs prefer to form a cluster of size $3$ for $13.1$\% of the time for $c=0.5$ and $\lambda_m=0.09$, whereas MTDs prefer to form a cluster of size $5$ for $8.2$\% of the time for $c=30$ and $\lambda_m=0.09$.  As the correlation constant increases within the cluster, MTDs prefer to form slightly larger clusters.  The results also show that, in the proposed evolutionary algorithm, MTDs prefer to form a cluster of size $3$ for $19.1$\% of the time for $\lambda_m=0.36$ and $c=6$, whereas MTDs prefer to form a cluster of size $5$ for $8$\% of the time for $\lambda_m=0.045$ and $c=6$.

\begin{figure*}[!t]
 \centering
 \begin{multicols}{2}
 \includegraphics[width=\linewidth]{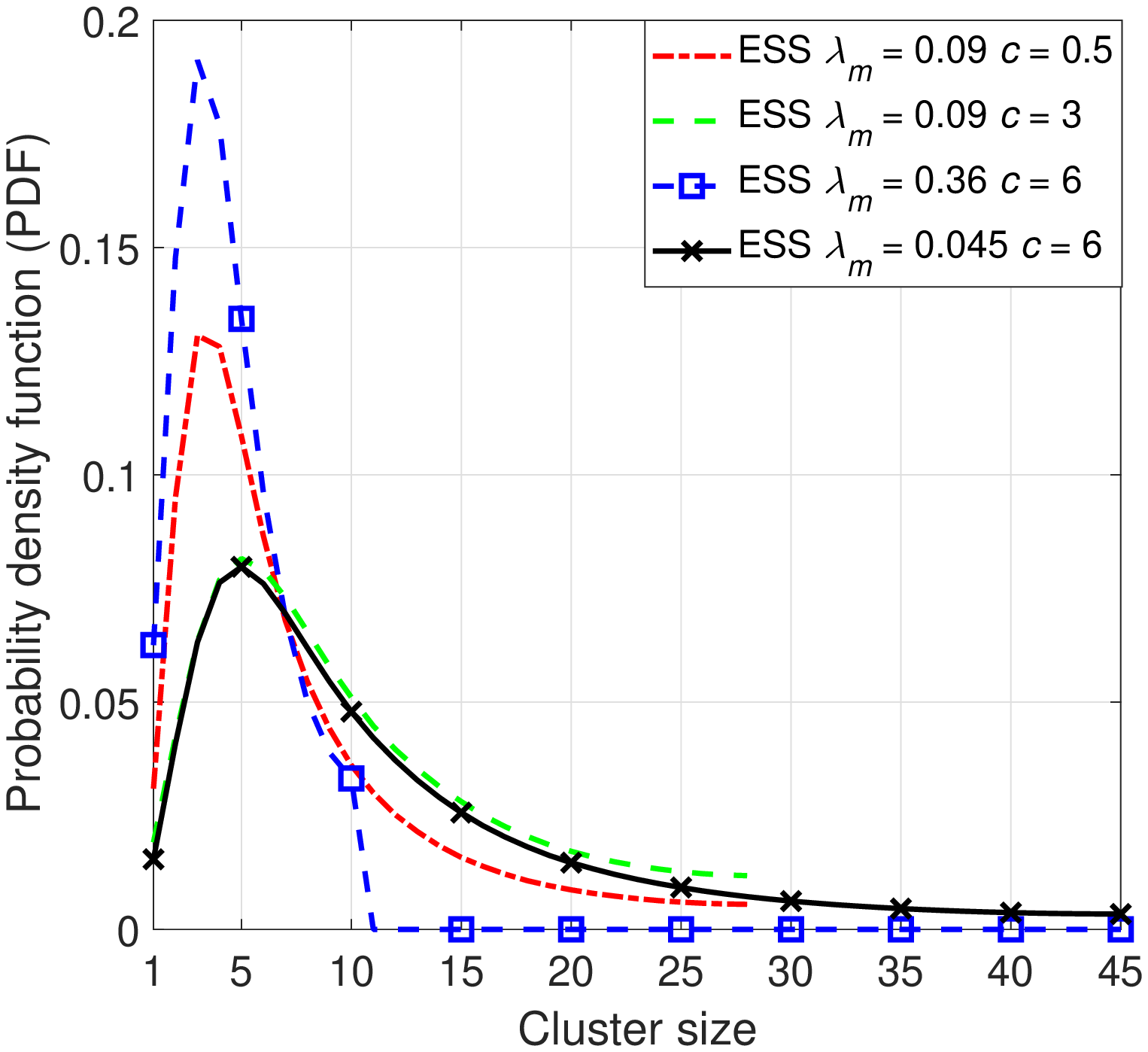}\par\vspace{-0.3cm}\caption{Percentage of MTDs forming a particular cluster size as a function of correlation constants, $c$, and MTD density, $\lambda_m$, where $\nu=2.5$.}\label{fig:R4}
 \includegraphics[width=\linewidth]{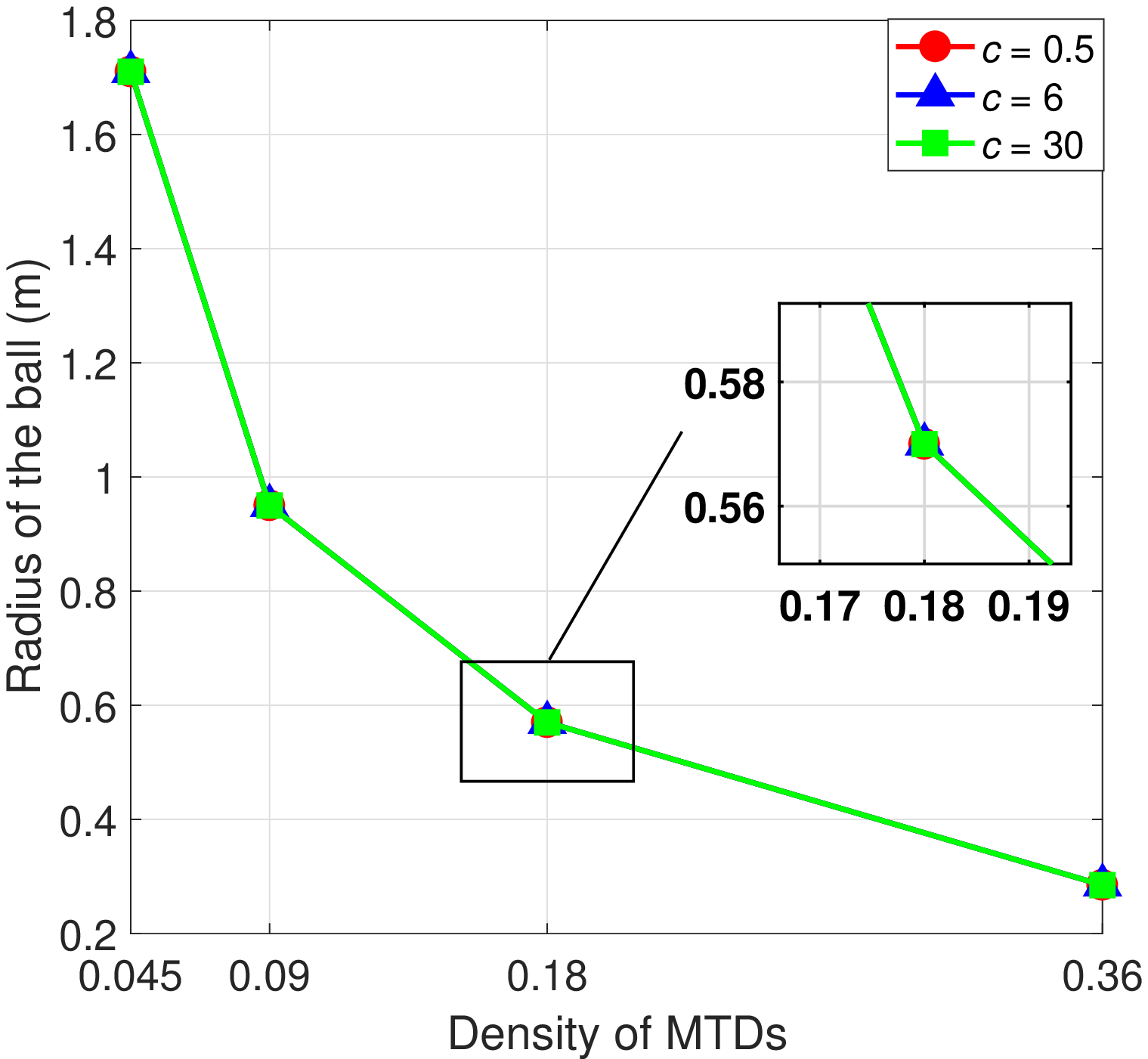}\par\vspace{-0.3cm}\caption{Radius of the ball as a function of MTD density, $\lambda_m$, and correlation constants, $c$, where $\nu=2.5$.}\label{fig:R12}
 \end{multicols}
 \vspace{-0.6cm}
\end{figure*}

Fig.~\ref{fig:R11} shows the average transmission power per MTD per cluster over time, under different correlation constants.  In this figure, we can see that any increase in the correlation constant will lead to a decrease in the average transmission power per MTD per cluster, due to the increase in data correlation within the cluster.  The proposed evolutionary algorithm converges to an ESS at $t=1000$ for all correlation constants.  Over time, the proposed evolutionary algorithm converges to an average transmit power per MTD of $11.5$~dBm for $c=0.5$, $11.18$~dBm for $c=6$, and $11.14$~dBm for $c=30$.

\begin{figure*}[!t]
 \centering
 \begin{multicols}{2}
 \includegraphics[width=\columnwidth]{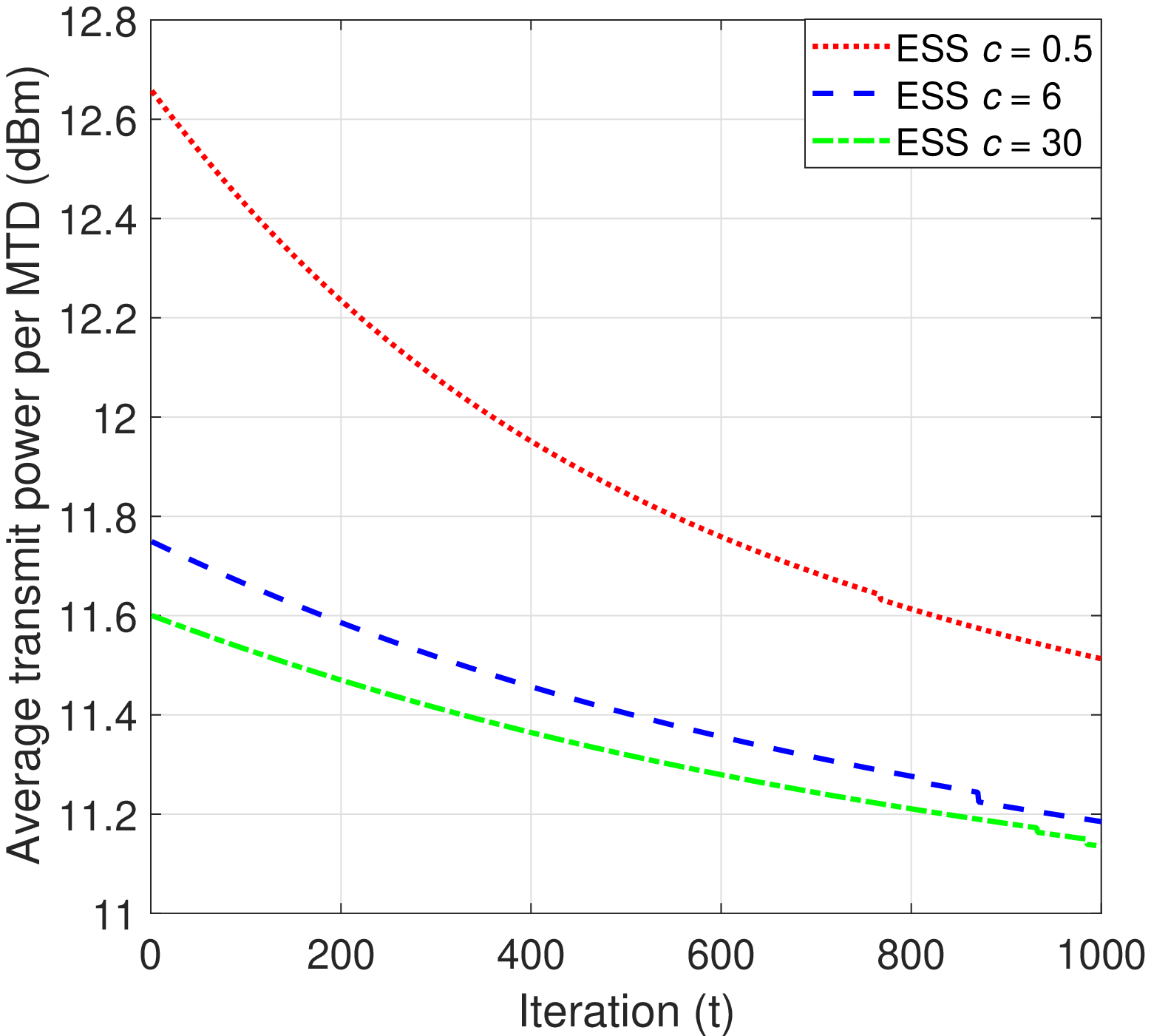}\par\vspace{-0.3cm}\caption{Dynamics of the average transmit power per MTD per cluster for different correlation constants, $c$, where density of MTDs is $\lambda_m=0.09$ and $\nu=2.5$.\vspace{-1cm}}\label{fig:R11}
  \includegraphics[width=\linewidth]{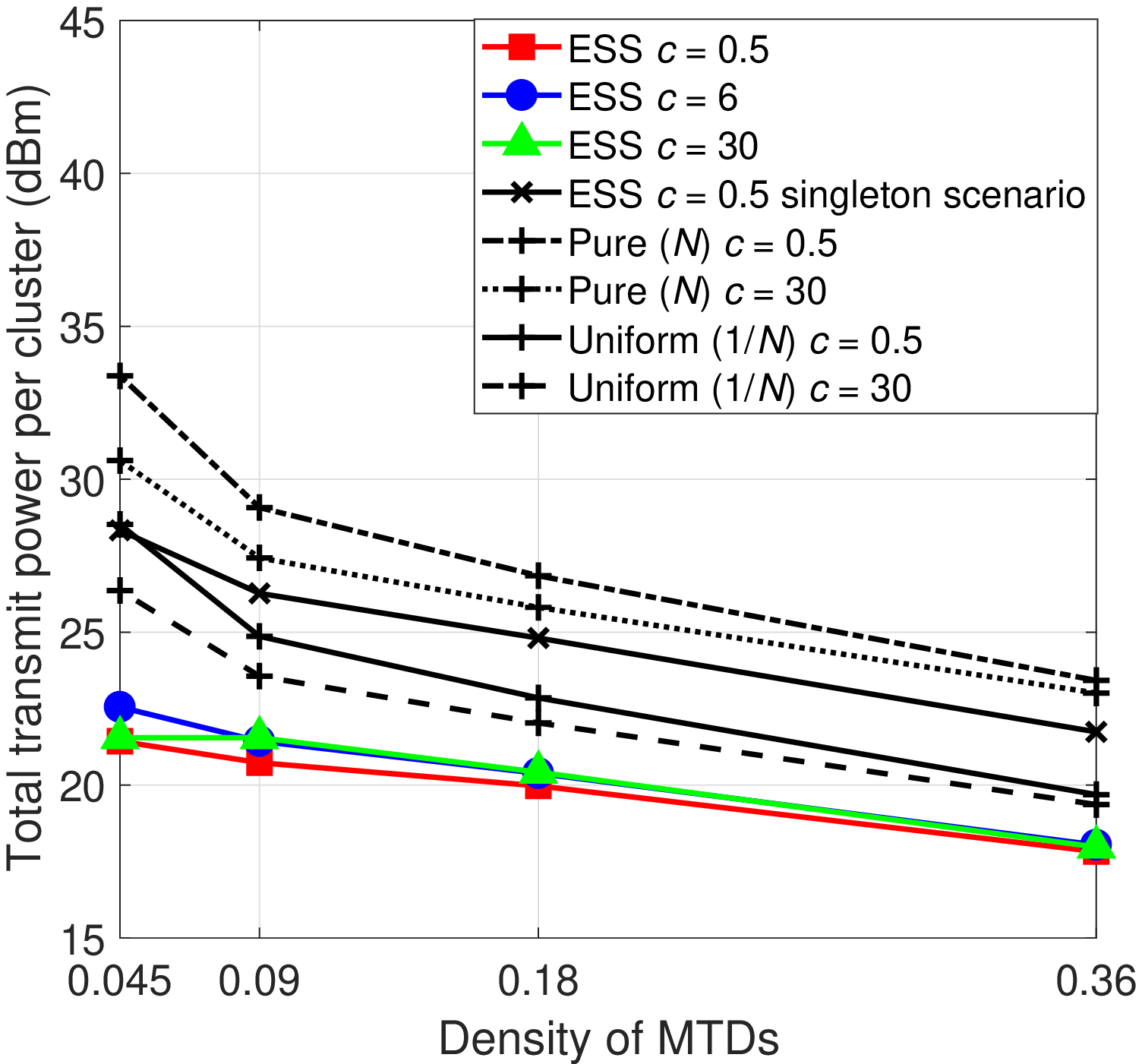}\par\vspace{-0.3cm}\caption{Total transmit power per cluster as a function of MTD density, $\lambda_m$, where $\nu=2.5$.}\label{fig:R2}
 \end{multicols}
\end{figure*}

Moreover, in Fig.~\ref{fig:R2}, we show the total transmit power per cluster at convergence, under different network densities and correlation constants.  We observe that, as the network density and correlation constant increase, the cluster size and data correlation will increase.  As cluster size and data correlation increase, the total transmit power per cluster will also increase.  We compare the proposed evolutionary-based correlation aware clustering algorithm for different correlation constants to the singleton scenario, as well as the two benchmarks.  The singleton scenario, assumes all MTDs within the ball transmit their own data via the cellular link to the BS, thus no clustering within the ball.  On average, the transmit power per cluster is reduced of around $20.7$\% for $c=0.5$, $18.4$\% for $c=6$, and $18.1$\% for $c=30$, compared to the singleton scenario.  The proposed evolutionary algorithm decreases the transmission power per cluster of around $28.5$\% for $c=0.5$, $26.4$\% for $c=6$, and $26.2$\% for $c=30$, compared to the pure cluster type benchmark for $c=0.5$, whereas transmission power per cluster decreases of around $24.9$\% for $c=0.5$, $22.7$\% for $c=6$, and $22.4$\% for $c=30$, compared to the pure cluster type benchmark for $c=30$.  On the other hand, the proposed evolutionary algorithm also decreases transmission power per cluster of around $15.9$\% for $c=0.5$, $13.5$\% for $c=6$, and $13.2$\% for $c=30$, compared to the uniform cluster type benchmark for $c=0.5$, whereas transmission power per cluster decreases around $12$\% for $c=0.5$, $9.5$\% for $c=6$, and about $9.2$\% for $c=30$, compared to the uniform cluster type benchmark for $c=30$.  On the average, the proposed evolutionary algorithm minimizes transmit power by $23.4$\% and $9.6$\% across all correlation constants and network densities, compared to the pure cluster type benchmark and the uniform cluster type benchmark respectively.

\begin{figure*}[!t]
 \centering
 \begin{multicols}{2}
 \includegraphics[width=\linewidth]{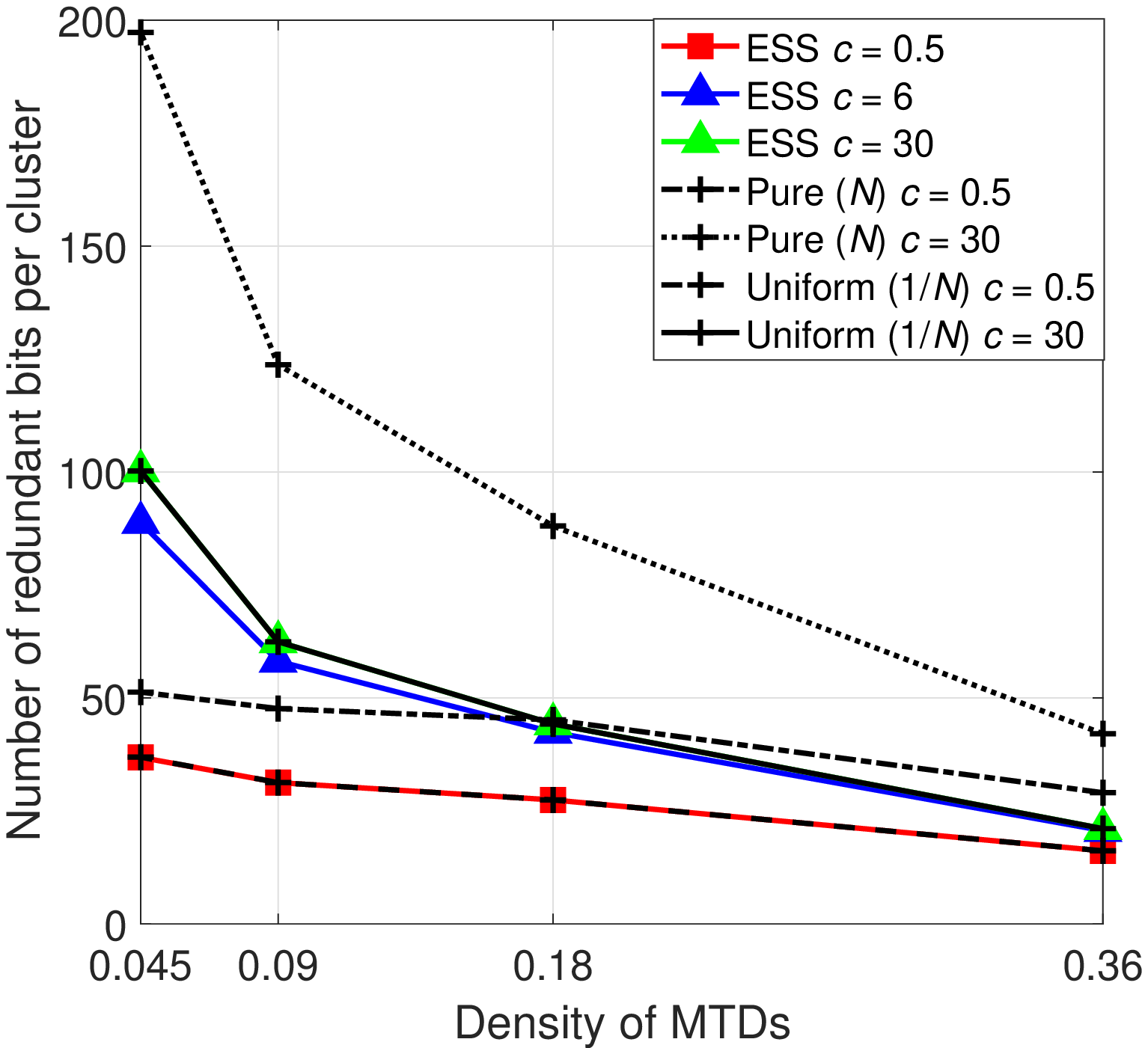}\par\vspace{-0.3cm}\caption{Number of redundant bits per cluster as a function of MTD density, $\lambda_m$, where $\nu=2.5$.}\label{fig:R9}
 \includegraphics[width=\linewidth]{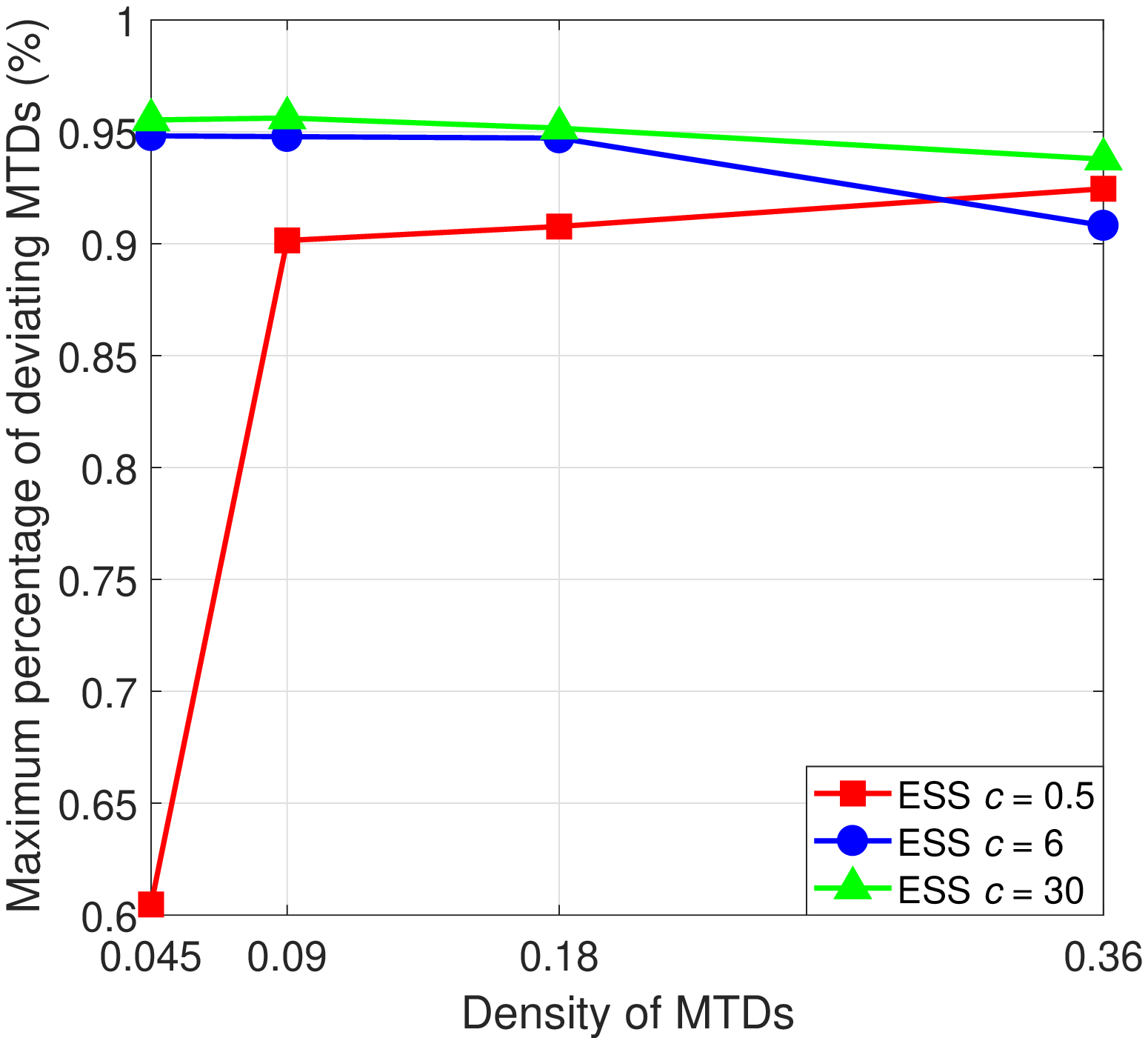}\par\vspace{-0.3cm}\caption{Maximum percentage of deviating MTDs from ESS as a function of MTD density, $\lambda_m$, where $\nu=2.5$.}\label{fig:R7}
 \end{multicols}
 \vspace{-1cm}
\end{figure*}

Fig.~\ref{fig:R9} shows the average number of redundant bits per cluster, for different network densities and correlation constants.  The average number of redundant bits per cluster, is the amount of bits that could be removed by the cluster head, before transmitting the cluster data to the BS.  Thus, decreasing network density and increasing the correlation constant, as in Fig.~\ref{fig:R9}, leads to an increase in data correlation, cluster size, and the number of redundant bits.  On the average, in Fig.~\ref{fig:R9}, the proposed evolutionary-based correlation aware clustering algorithm increases the number of redundant bits by more than double, that is, $43.6$~\%, when increasing the correlation constant from $c=0.5$ to $c=30$. On the average, the proposed evolutionary correlation-aware clustering algorithm percentage comparisons to the pure cluster type benchmark and uniform cluster type benchmark, are outlined in Table~\ref{tab:sim2}.  As shown in Fig.~\ref{fig:R9} and Table~\ref{tab:sim2}, the pure benchmark with correlation constant $c=30$ maximizes the number of redundant bits per cluster across all network densities, as this benchmark considers all MTDs preferring to form the maximum possible size cluster within the ball.  Thus, considering the results of Figs.~\ref{fig:R11} and \ref{fig:R2}, both benchmarks yielded a larger transmit power and number of redundant bits per cluster, compared to the proposed distributed algorithm.  However, both benchmarks require a centralized approach in order to achieve these results, and those benchmarks are not robust to stochastic changes in the M2M network environment.

\begin{table} [!t]
\vspace{-1cm}
\begin{center}
\caption{Simulation analysis for Fig.~\ref{fig:R9}\vspace{-0.3cm}}\label{tab:sim2}
\begin{tabular}{|c|c|c|c|c|c|}
\hline
\textbf{ESS Algorithm} & \textbf{Pure} $c=0.5$ & \textbf{Pure} $c=30$ & \textbf{Uniform} $c=0.5$ & \textbf{Uniform} $c=30$ \\ \hline
$c=0.5$ &  $-36.5$\% & $-71.6$\% & $-0.04$\% & $-43.6$\%  \\ \hline
$c=6$   &  $15$\%    & $-52.7$\% & $77.2$\%  & $-6.1$\%  \\ \hline
$c=30$  &  $24.3$\%  & $-49.6$\% & $90.8$\%  & $0$\%  \\
\hline
\end{tabular}
\end{center}
\vspace{-0.3cm}
\end{table}

In Fig.~\ref{fig:R7}, we show the maximum percentage of deviating MTDs at an ESS for our proposed evolutionary game, under different network densities, $\lambda_m$, and correlation constants, $c$.  As the density and correlation constant increase (that is, as the network becomes denser and data in the cluster becomes more correlated) the maximum percentage of deviating MTDs also increases.  For a correlation constant, $c=0.5$, when the network density is $\lambda_m=0.045$ the maximum portion of MTDs that will deviate due to additional deployment of MTDs or MTD loss of battery, is about $0.6$\%.  However, as the density of the network increases to $\lambda_m=0.36$, the maximum percentage of MTDs that will deviate is about $0.92$\%.  Therefore, at the ESS, as the network becomes denser and more correlated, MTDs have more options to deviate compared to when the network density is sparse.  Additionally, when the network becomes more correlated, the maximum percentage of MTDs that may deviate saturates when the network density increases.

Fig.~\ref{fig:R10} shows the average transmit power per cluster as a function of cluster size, for different network densities.  We observe an increasing relationship between cluster size and transmit power per cluster.  Depending on network density, the potential cluster size changes, which is due to the changing of the ball radius with respect to the network density.  As the network density decreases, the ball radius will also decrease, leading to an increase in the potential number of MTDs per cluster.  On average, in Fig.~\ref{fig:R10}, the proposed evolutionary-based correlation aware clustering algorithm has a maximum transmit power per cluster of $31.2$~dBm for a maximum cluster size of $45$ for $\lambda_m=0.045$, whereas for $\lambda_m=0.18$ the proposed game has a maximum transmit power per cluster of $25.9$~dBm for a maximum cluster size of $20$.

\begin{figure*}[!t]
 \centering
 \begin{multicols}{2}
 \includegraphics[width=\linewidth]{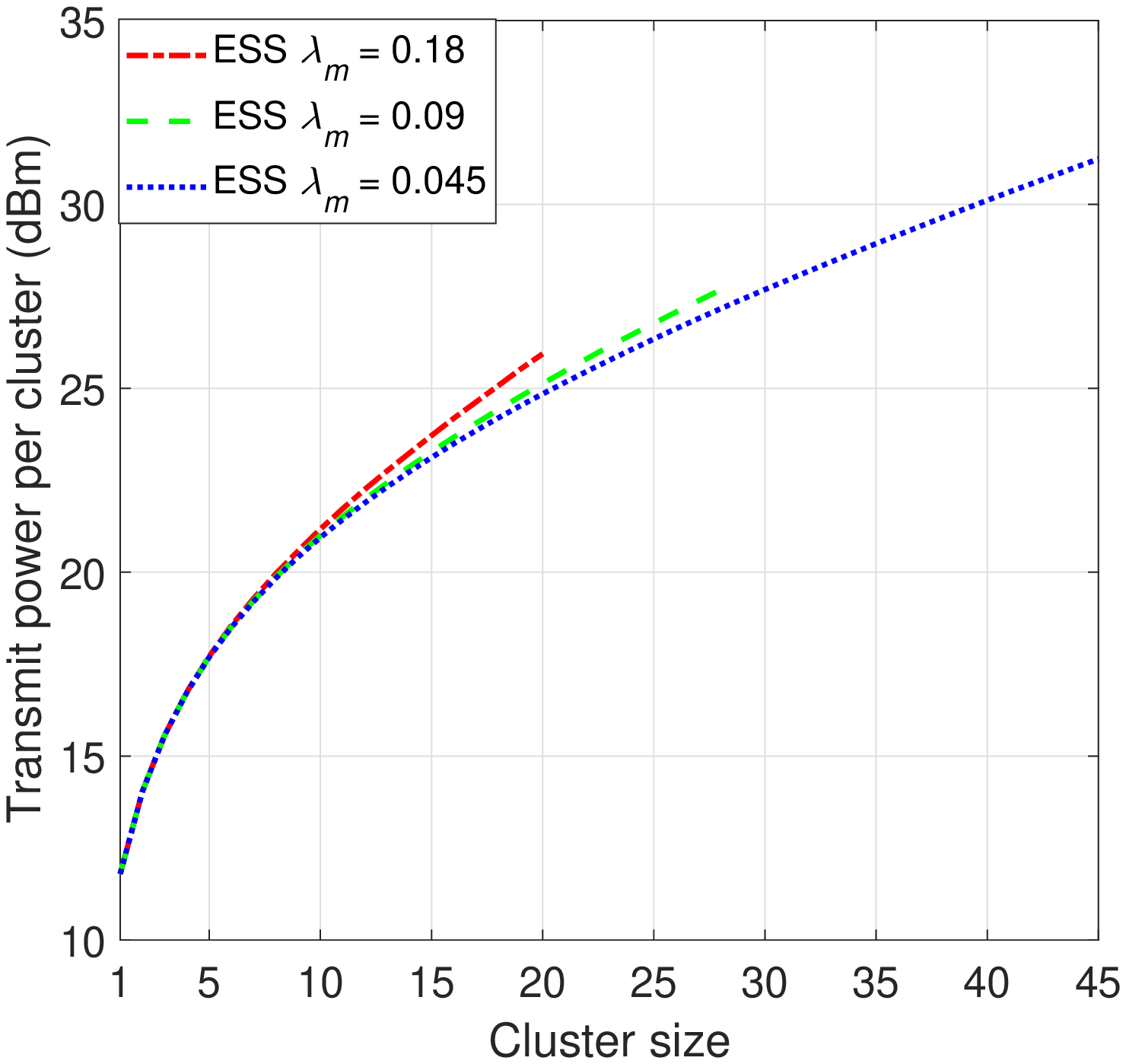}\par\vspace{-0.3cm}\caption{Total transmit power per cluster as a function of potential cluster sizes, MTD density, $\lambda_m$, $\nu=2.5$, and $c=6$.}\label{fig:R10}
 \includegraphics[width=\linewidth]{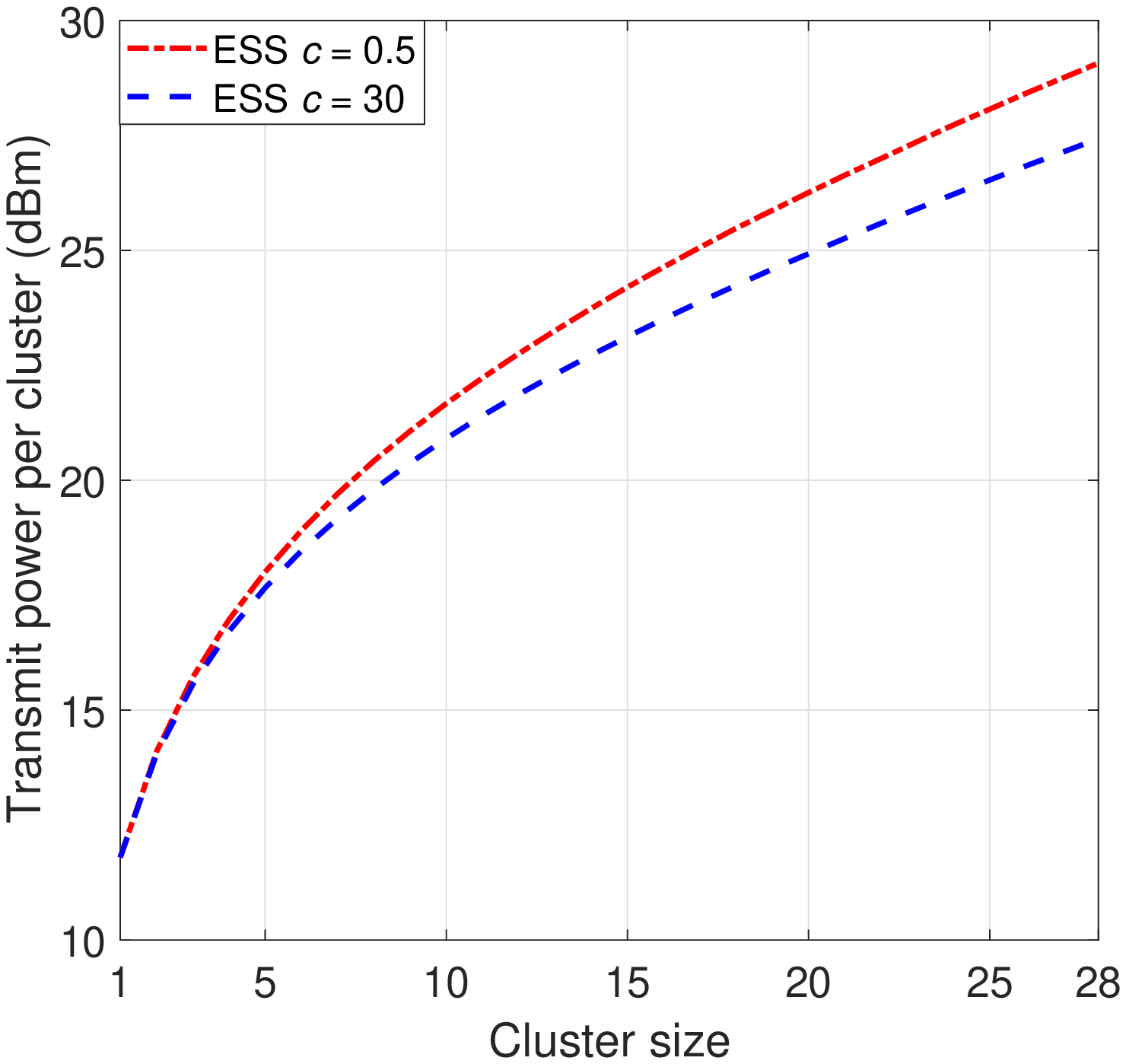}\par\vspace{-0.3cm}\caption{Total transmit power per cluster as a function of potential cluster sizes, correlation constant, $c$, $\nu=2.5$, and $\lambda_m=0.09$.}\label{fig:R13}
 \end{multicols}
 \vspace{-1cm}
\end{figure*}

Fig.~\ref{fig:R13} shows the average transmit power per cluster as a function of cluster size, for different correlation constants.  We observe an increasing relationship between cluster size and transmit power per cluster size, that is similar to the one shown in Fig.~\ref{fig:R10}.  We observe that, as the correlation constant increases, the potential cluster sizes do not change, which is again due to the ball radius adapting to the changes in the network.  On average, in Fig.~\ref{fig:R13}, the proposed evolutionary algorithm has a maximum transmit power per cluster of $29.1$~dBm for a maximum cluster size of $28$ for $c=0.5$, whereas for $c=30$ the proposed game has a maximum transmit power per cluster of $27.4$~dBm for a maximum cluster size of $28$.

\vspace{-0.1cm}
\section{Conclusion}\label{Sec:Con}
In this paper, we have proposed a novel distributed correlation-aware clustering scheme for reducing the number of redundant bits being set to the BS, as well as reducing transmission power for each MTD in a massive and locally finite M2M network.  In the proposed model, MTDs self-organize in a fully distributed manner to form cluster, based on data correlation and potential transmission power savings.  We have modeled the problem using evolutionary game theory and stochastic geometry.  Stochastic geometry has been used to accurately model and characterize the distance distributions between MTDs, in order to derive a closed-form expression for inter-cluster interference.  Additionally, stochasic geometry has also been used to derive average power consumption per cluster, as a function of MTD location, cluster size, and inter-cluster interference.  Based on the evolutionary game, we have proposed a distributed clustering algorithm, in which MTDs autonomously seek to minimize the average MTD transmission power per cluster.  We have shown that the proposed distributed algorithm converges to a stable state, that is, an ESS which is robust to a small portion of MTDs deviating from the stable cluster formation.  Moreover, we have also derived a maximum portion of MTDs that can deviate from the ESS, while still maintaining a stable cluster formation.  Simulation results show that, the average, the proposed evolutionary algorithm decreases transmission power per cluster, and increases the number of redundant bits that can be eliminated in a given cluster.

\appendix
\section{}
\subsection{Proof of Theorem~\ref{Theor_utility}}\label{app:A}
The utility function $u_{\mathcal{K}_{i,j}}(t)$ of typical cluster $\mathcal{K}_{i,j}$ choosing type $j$, is expressed as follows:
\begin{align}\label{Utility-T}
u_{\mathcal{K}_{i,j}}(t)=-\frac{1}{j}\sum_{k=1}^{j-1} \int_{0}^{r_j} \left(\frac{\mathbb{E}(I_j)}{g_{ki}}\right)\left(2^{\frac{H_k}{TB}}-1\right) f_{\mathcal{D}_i}^{(k)}(r)dr-\frac{1}{j}\left(\frac{BN_0}{g_{i}}\right)\left(2^{\frac{H_{\mathcal{K}_{i,j}}}{TB}}-1\right).
\end{align}

In (\ref{Utility-T}), there is a tradeoff between cluster size and transmission power for MTDs.  Thus, as the cluster size increases, the transmit power per cluster also increases, due to the increasing number of MTDs.  The proposed utility function in (\ref{Utility-T}) captures the average transmission power for all MTDs across the M2M and cellular links within the cluster, sending all the gathered data of the cluster.  From (\ref{jointEntropy}), we can rewrite (\ref{Utility-T}) as follows:
\begin{align}\label{Utility-K}
 u_{\mathcal{K}_{i,j}}(t)=&-\frac{1}{j}\sum_{k=1}^{j-1} \int_{0}^{r_j} \left(\mathbb{E}(I_j)
 \left(\frac{4\pi r}{\mu}\right)^{\nu}
 \right)
 \left(2^{\frac{H_k}{TB}}-1\right) f_{\mathcal{D}_i}^{(k)}(r)dr- \nonumber \\
 &\frac{1}{j}\left(\frac{BN_0}{g_{i}}\right)\left(2^{\frac{H_m+H_m(j-1)\big(1-\frac{\alpha}{\frac{r_j}{c}+1}\big)}{TB}}-1\right).
\end{align}

The utility function in (\ref{Utility-K}) for a given cluster $\mathcal{K}_{i,j}$, is a function of the cluster radius $r_j$ and the cluster type $j$, where $j$ is equal to the cluster size, $j=K_{i,j}$.  Furthermore, from (\ref{K_Distance}) and (\ref{Interferense_simple}), we can rewrite the utility function (\ref{Utility-K}) of cluster $\mathcal{K}_{i,j}$ choosing type $j$ as a closed-form expression, as follows:
\begin{align}\label{Utility-K-closed_1}
u_{\mathcal{K}_{i,j}}(t)=&-
 \frac{1}{j}\left(\frac{4\pi}{\mu}\right)^{\nu}
 \mathbb{E}(I_j) \left(2^{\frac{H_m}{TB}}-1\right)\sum_{k=1}^{j-1} \frac{(j-1)!}{(k-1)! (j-1-k)!}\int_{0}^{r_j}
\frac{2r^{\nu+1}}{r_j^2}\left(\frac{r^2}{r_j^2}\right)^{k-1} \times \nonumber \\
& \left(1-{\frac{r^2}{r_j^2}}\right)^{j-1-k}dr-\frac{1}{j}\left(\frac{BN_0}{g_{i}}\right)\left(2^{\frac{H_m+H_m(j-1)\left(1-\frac{\alpha}{\frac{r_j}{c}+1}\right)}{TB}}-1\right).
\end{align}

To further simplify (\ref{Utility-K-closed_1}), if $\frac{r^2}{r_j^2}=a$, then $dr=\frac{r_j da}{2 a^{\frac{1}{2}}}$ and $r=r_j a^{\frac{1}{2}}$. Thus:
\begin{align}
\int_{0}^{r_j} \frac{2r^{\nu+1}}{r_j^2}\left(\frac{r^2}{r_j^2}\right)^{k-1} \left(1-{\frac{r^2}{r_j^2}}\right)^{j-1-k}dr & =\int_{0}^{1}
2 a r_j^{\nu-1} a^{\frac{\nu-1}{2}} a^{k-1} \left(1-{a}\right)^{j-1-k} \frac{r_j da}{2 a^{\frac{1}{2}}} \\ \nonumber
\therefore &= r_j^\nu \int_{0}^{1} a^{k+\frac{\nu-2}{2}} \left(1-{a}\right)^{j-1-k} da. \nonumber
\end{align}
Since $\int_{0}^{r_j} \frac{2r^{\nu+1}}{r_j^2}\left(\frac{r^2}{r_j^2}\right)^{k-1}\left(1-{\frac{r^2}{r_j^2}}\right)^{j-1-k}dr =\frac{r_j^\nu\Gamma(j-k)\Gamma(k+\frac{v}{2})}{\Gamma(j+\frac{v}{2})}$, where $\Gamma(\cdot)$ is the Gamma function, and $j$ and $k$ are integers, such that $j>k$, and $2k+\nu>0$ where $\nu>2$ and $\nu$ is not an integer.  Thus, (\ref{Utility-K-closed_1}) can be rewritten as:
\begin{align}\label{Utility-K-closed_2}
u_{\mathcal{K}_{i,j}}(t)&=-\frac{1}{j}\left(\frac{4\pi}{\mu}\right)^{\nu}\mathbb{E}(I_j)\left(2^{\frac{H_m}{TB}}-1\right)\sum_{k=1}^{j-1} \frac{(j-1)!}{(k-1)! (j-1-k)!}  \frac{r_j^\nu\Gamma(j-k)\Gamma(k+\frac{\nu}{2})}{\Gamma(j+\frac{\nu}{2})}-\nonumber \\
&\frac{1}{j}\left(\frac{BN_0}{g_{i}}\right)\left(2^{\frac{H_m+H_m(j-1)\left(1-\frac{\alpha}{\frac{r_j}{c}+1}\right)}{TB}}-1\right).
\end{align}

Since path loss, $\nu$, is not always an integer we cannot express (\ref{Utility-K-closed_2}) as a function of factorials, thus we can further simplify (\ref{Utility-K-closed_2}) in terms of the $\Gamma$ function, given that $\Gamma(n)=(n-1)!$, $\Gamma(n+1)=n\Gamma(n)$, and $(n-1-m)!=\Gamma(n-m)$, where $n$ and $m$ are positive integers.  (\ref{Utility-K-closed_2}) can be rewritten as:
\begin{align}\label{Utility-K-closed_3}
u_{\mathcal{K}_{i,j}}(t)=&-\frac{1}{j}\left(\frac{4\pi r_j}{\mu}\right)^{\nu}\mathbb{E}(I_j)\left(2^{\frac{H_m}{TB}}-1\right)\sum_{k=1}^{j-1}\frac{\Gamma(j)\Gamma(k+\frac{\nu}{2})}{\Gamma(k)\Gamma(j+\frac{\nu}{2})}- \nonumber \\
&\frac{1}{j}\left(\frac{BN_0}{g_{i}}\right)\left(2^{\frac{H_m+H_m(j-1)\left(1-\frac{\alpha}{\frac{r_j}{c}+1}\right)}{TB}}-1\right).
\end{align}

Therefore, by deriving the closed-form expression of the summation in (\ref{Utility-K-closed_3}), $\sum_{k=1}^{j-1} \frac{\Gamma(j)\Gamma(k+\frac{\nu}{2})}{\Gamma(k)\Gamma(j+\frac{\nu}{2})}=\frac{2\left(j-1\right)}{2+\nu}$, we can find the closed-form expression of (\ref{Utility-T}), as defined in (\ref{Utility-closed}).

\subsection{Proof of Theorem~\ref{Theor_band_epsilon}}\label{app:B}
Since the set of cluster types in game $\mathcal{G}_E$ is finite,  the replicator dynamic will converge to the steady state which is the (symmetric) mixed strategy Nash equilibrium of $\mathcal{G}_E$~\cite{shoham2008multiagent}. Let $\boldsymbol{x}^*$ be the converged population state of the replicator dynamic from  proposed Algorithm~\ref{algorithm-E}. Following~(\ref{xdot}), at the converged population state of proposed Algorithm~\ref{algorithm-E}, we have:
\begin{equation}\label{Jacobi:1}
\frac{\partial \dot{x}_j^*(t)}{\partial x_m^*(t)}  =
  \begin{cases}
  \bar{u}_{\mathcal{K}_{i,j}}(\boldsymbol{x}^*,t)-U_{\mathcal{K}_i}(\boldsymbol{x}^*,t)+
  \left(
\frac{\partial\bar{u}_{\mathcal{K}_{i,j}}(\boldsymbol{x}^*,t)}{\partial x_j^*(t)}-
\frac{\partial U_{\mathcal{K}_i}(\boldsymbol{x}^*,t)}{\partial x_j^*(t)}
  \right)  x_j^*, & \text{for } j=m \\
     \left(
\frac{\partial\bar{u}_{\mathcal{K}_{i,j}}(\boldsymbol{x}^*,t)}{\partial x_m^*(t)}-
\frac{\partial U_{\mathcal{K}_i}(\boldsymbol{x}^*,t)}{\partial x_m^*(t)}
  \right)  x_m^*
   , & \text{for } j\neq m
  \end{cases}
\end{equation}

Based on (\ref{xdot}), at the convergence, for any $0 < x_j^*(t)$, $\bar{u}_{\mathcal{K}_{i,j}}(\boldsymbol{x}^*,t)-U_{\mathcal{K}_i}(\boldsymbol{x}^*,t)=0$. Moreover,
$\frac{\partial\bar{u}_{\mathcal{K}_{i,j}}(\boldsymbol{x}^*,t)}{\partial x_j^*(t)}=u_{\mathcal{K}_{i,j}}$,
$\frac{\partial U_{\mathcal{K}_i}(\boldsymbol{x}^*,t)}{\partial x_j^*(t)}=
\bar{u}_{\mathcal{K}_{i,j}}(\boldsymbol{x},t)+u_{\mathcal{K}_{i,j}}(\boldsymbol{x},t)x_j^*+
\sum_{n \in \mathcal{S} \setminus \{j\} }
u_{\mathcal{K}_{i,nj}}(\boldsymbol{x},t)x_n^*
$,
$\frac{\partial\bar{u}_{\mathcal{K}_{i,j}}(\boldsymbol{x}^*,t)}{\partial x_m^*(t)}=u_{\mathcal{K}_{i,jm}}$, and
$\frac{\partial U_{\mathcal{K}_i}(\boldsymbol{x}^*,t)}{\partial x_m^*(t)}=
\bar{u}_{\mathcal{K}_{i,m}}(\boldsymbol{x},t)+u_{\mathcal{K}_{i,jm}}(\boldsymbol{x},t)x_m^*+
\sum_{n \in \mathcal{S} \setminus \{m\} }
u_{\mathcal{K}_{i,jn}}(\boldsymbol{x},t)x_n^*
$
. Therefore, (\ref{Jacobi:1}) can be rewritten as follows:
\begin{equation}\label{Jacobi:2}
\frac{\partial \dot{x}_j^*(t)}{\partial x_m^*(t)}  =
  \big(u_{\mathcal{K}_{i,jm}}(\boldsymbol{x}^*,t)-2\bar{u}_{\mathcal{K}_{i,m}}(\boldsymbol{x}^*,t)\big)x_m^*.
\end{equation}

Since $\boldsymbol{x}^*$ is mixed strategy Nash equilibrium, i.e.,
$u_{\mathcal{K}_{i,m}}(\boldsymbol{x}^*,t) \leq \sum_{j^\prime\in\mathcal{S}}u_{\mathcal{K}_{i,mj^\prime}}(t)x_{j^\prime}^*$, and based on (\ref{fitness}) $\bar{u}_{\mathcal{K}_{i,m}}(\boldsymbol{x}^*,t)=\sum_{j^\prime\in\mathcal{S}}u_{\mathcal{K}_{i,mj^\prime}}(t)x_{j^\prime}^*$, and thus, $\frac{\partial \dot{x}_j^*(t)}{\partial x_m^*(t)}<0$. Consequently, the convergence state of the proposed algorithm, and specifically the steady state of the replicator dynamic in (\ref{xdot}) is asymptotically stable. Since steady state of the replicator dynamic is an asymptotically stable state in our problem, then the proposed Algorithm~\ref{algorithm-E} will converge to ESS~\cite{shoham2008multiagent}. This completes the proof.

We assume $\boldsymbol{x}^{\epsilon}$ is a preference profile vector, where $\epsilon$ potion of MTDs deviate while the remaining $1-\epsilon$ potion of MTDs choose the converged population state $\boldsymbol{x}^{*}$. Let $ x_j^{\epsilon}=x_{j}^{*}+\epsilon_{j} \frac{\partial \dot{x}_{j}}{\partial x_{j}}(\boldsymbol{x}^*)$. Since, $0\leq x_{j}^{\epsilon}$, we have:
\begin{align}
& |\epsilon_j|\leq 1, \forall j\in \mathcal{S},\\
&\sum_{j\in \mathcal{S}} \epsilon_j=0,\\
&\therefore|\epsilon_{j}|\leq
  \frac{1}{2\bar{u}_{\mathcal{K}_{i,j}}(\boldsymbol{x}^*,t)-u_{\mathcal{K}_{i,j}}(\boldsymbol{x}^*,t)} \forall j\in \mathcal{S}.
\end{align}

To hold the above inequalities, the maximum value of $\epsilon$ is given by ($\ref{Band_epsilon}$).

\vspace{-0.05cm}
\bibliographystyle{IEEEtran}
\def\baselinestretch{0.9}


\begin{thebibliography}{10}
\providecommand{\url}[1]{#1}
\csname url@samestyle\endcsname
\providecommand{\newblock}{\relax}
\providecommand{\bibinfo}[2]{#2}
\providecommand{\BIBentrySTDinterwordspacing}{\spaceskip=0pt\relax}
\providecommand{\BIBentryALTinterwordstretchfactor}{4}
\providecommand{\BIBentryALTinterwordspacing}{\spaceskip=\fontdimen2\font plus
\BIBentryALTinterwordstretchfactor\fontdimen3\font minus
  \fontdimen4\font\relax}
\providecommand{\BIBforeignlanguage}[2]{{%
\expandafter\ifx\csname l@#1\endcsname\relax
\typeout{** WARNING: IEEEtran.bst: No hyphenation pattern has been}%
\typeout{** loaded for the language `#1'. Using the pattern for}%
\typeout{** the default language instead.}%
\else
\language=\csname l@#1\endcsname
\fi
#2}}
\providecommand{\BIBdecl}{\relax}
\BIBdecl

\bibitem{sawyer2017evoluationary}
N.~Sawyer, M.~N. Soorki, W.~Saad, and D.~B. Smith, ``Evolutionary coalitional
  game for correlation-aware clustering in machine-to-machine communications,''
  in \emph{Proc.\ of IEEE Global Communications Conference (GLOBECOM)},
  Singapore, Dec 2017, pp. 1--6.

\bibitem{chen2014machine}
K.-C. Chen and S.-Y. Lien, ``Machine-to-machine communications: {T}echnologies
  and challenges,'' \emph{Ad Hoc Networks}, vol.~18, pp. 3--23, July 2014.

\bibitem{zheng2012radio}
K.~Zheng, F.~Hu, W.~Wang, W.~Xiang, and M.~Dohler, ``Radio resource allocation
  in {LTE}-{A}dvanced cellular networks with {M2M} communications,'' \emph{IEEE
  Communications Magazine}, vol.~50, no.~7, July 2012.

\bibitem{verma2016machine}
P.~K. Verma, R.~Verma, A.~Prakash, A.~Agrawal, K.~Naik, R.~Tripathi,
  M.~Alsabaan, T.~Khalifa, T.~Abdelkader, and A.~Abogharaf,
  ``Machine-to-machine ({M2M}) communications: {A} survey,'' \emph{Journal of
  Network and Computer Applications}, vol.~66, pp. 83--105, May 2016.

\bibitem{mozaffari2016unmanned}
M.~Mozaffari, W.~Saad, M.~Bennis, and M.~Debbah, ``Unmanned aerial vehicle with
  underlaid device-to-device communications: Performance and tradeoffs,''
  \emph{IEEE Trans. on Wireless Communications}, vol.~15, no.~6, pp.
  3949--3963, June 2016.

\bibitem{chen2017caching}
M.~Chen, M.~Mozaffari, W.~Saad, C.~Yin, M.~Debbah, and C.~S. Hong, ``Caching in
  the sky: Proactive deployment of cache-enabled unmanned aerial vehicles for
  optimized quality-of-experience,'' \emph{IEEE Journal on Selected Areas in
  Communications (JSAC)}, vol.~35, no.~5, pp. 1046--1061, May 2017.

\bibitem{park2016learning}
T.~Park, N.~Abuzainab, and W.~Saad, ``Learning how to communicate in the
  {I}nternet of {T}hings: Finite resources and heterogeneity,'' \emph{IEEE
  Access, Special Issue on Optimization for Emerging Wireless Networks: IoT, 5G
  and Smart Grid Communication Networks}, vol.~4, pp. 7063--7073, Nov 2016.

\bibitem{dawy2017toward}
Z.~Dawy, W.~Saad, A.~Ghosh, J.~G. Andrews, and E.~Yaacoub, ``Toward massive
  machine type cellular communications,'' \emph{IEEE Wireless Communications},
  vol.~24, no.~1, pp. 120--128, Feb 2017.

\bibitem{lien2011toward}
S.-Y. Lien, K.-C. Chen, and Y.~Lin, ``Toward ubiquitous massive accesses in
  3{GPP} machine-to-machine communications,'' \emph{IEEE Communications
  Magazine}, vol.~49, no.~4, April 2011.

\bibitem{abuzainab2017cognitive}
N.~Abuzainab, W.~Saad, C.-S. Hong, and H.~V. Poor, ``Cognitive hierarchy theory
  for distributed resource allocation in the {I}nternet of {T}hings,''
  \emph{IEEE Trans. on Wireless Communications}, vol.~16, no.~12, pp.
  7687--7702, Dec 2017.

\bibitem{juan2013data}
T.-C. Juan, S.-E. Wei, and H.-Y. Hsieh, ``Data-centric clustering for data
  gathering in machine-to-machine wireless networks,'' in \emph{Proc.\ of IEEE
  International Conference on Communications (ICC) Workshops}, Budapest,
  Hungary, June 2013, pp. 89--94.

\bibitem{song2015correlation}
H.~Song, H.-Y. Hsieh, Y.-D. Tsai, and W.~Choi, ``Correlation-aware machine
  selection for {M}2{M} data gathering in cellular networks,'' in \emph{Proc.\
  of IEEE 26th Annual International Symposium on Personal, Indoor, and Mobile
  Radio Communications (PIMRC)}, Hong Kong, China, Sep 2015, pp. 1184--1189.

\bibitem{miao20162}
G.~Miao, A.~Azari, and T.~Hwang, ``{E}$^{2}$-{MAC}: {E}nergy efficient medium
  access for massive {M2M} communications,'' \emph{IEEE Trans. on
  Communications}, vol.~64, no.~11, pp. 4720--4735, Nov 2016.

\bibitem{lee2013feasibility}
H.-K. Lee, D.~M. Kim, Y.~Hwang, S.~M. Yu, and S.-L. Kim, ``Feasibility of
  cognitive machine-to-machine communication using cellular bands,'' \emph{IEEE
  Wireless Communications}, vol.~20, no.~2, pp. 97--103, April 2013.

\bibitem{soorki2017stochastic}
M.~N. Soorki, W.~Saad, M.~H. Manshaei, and H.~Saidi, ``Stochastic coalitional
  games for cooperative random access in {M2M} communications,'' \emph{IEEE
  Trans. on Wireless Communications}, vol.~16, no.~9, pp. 6179--6192, Sep 2017.

\bibitem{ho2012energy}
C.~Y. Ho and C.-Y. Huang, ``Energy-saving massive access control and resource
  allocation schemes for {M2M} communications in {OFDMA} cellular networks,''
  \emph{IEEE Wireless Communications Letters}, vol.~1, no.~3, pp. 209--212,
  June 2012.

\bibitem{al2016self}
B.~R. Al-Kaseem, A.~O. Nyanteh, and H.~S. Al-Raweshidy, ``Self-organized
  clustering technique based on sink mobility in heterogeneous {M2M} sensor
  networks,'' in \emph{Proc.\ of International Conference for Students on
  Applied Engineering (ICSAE)}, Newcastle upon Tyne, UK, Oct 2016, pp.
  431--436.

\bibitem{wei2012joint}
S.-E. Wei, H.-Y. Hsieh, and H.-J. Su, ``Joint optimization of cluster formation
  and power control for interference-limited machine-to-machine
  communications,'' in \emph{Proc.\ of IEEE Global Communications Conference
  (GLOBECOM)}, Anaheim, CA, USA, Dec 2012, pp. 5512--5518.

\bibitem{hussain2014multi}
F.~Hussain, A.~Anpalagan, and M.~Naeem, ``Multi-objective {MTC} device
  controller resource optimization in {M2M} communication,'' in \emph{Proc.\ of
  27th Biennial Symposium on Communications (QBSC)}, Kingston, ON, Canada, June
  2014, pp. 184--188.

\bibitem{safdar2013resource}
H.~Safdar, N.~Fisal, R.~Ullah, W.~Maqbool, F.~Asraf, Z.~Khalid, and A.~Khan,
  ``Resource allocation for uplink {M2M} communication: A game theory
  approach,'' in \emph{Proc.\ of IEEE Symposium on Wireless Technology and
  Applications (ISWTA)}, Kuching, Malaysia, Sep 2013, pp. 48--52.

\bibitem{bayat2014distributed}
S.~Bayat, Y.~Li, Z.~Han, M.~Dohler, and B.~Vucetic, ``Distributed data
  aggregation in machine-to-machine communication networks based on coalitional
  game,'' in \emph{Proc.\ of Wireless Communications and Networking Conference
  (WCNC)}, Istanbul, Turkey, April 2014, pp. 2026--2031.

\bibitem{saad2009distributed}
W.~Saad, Z.~Han, M.~Debbah, and A.~Hjorungnes, ``A distributed coalition
  formation framework for fair user cooperation in wireless networks,''
  \emph{IEEE Trans. on wireless communications}, vol.~8, no.~9, pp. 4580--4593,
  Sep 2009.

\bibitem{hsieh2015not}
H.-Y. Hsieh, C.-H. Chang, and W.-C. Liao, ``Not every bit counts: Data-centric
  resource allocation for correlated data gathering in machine-to-machine
  wireless networks,'' \emph{ACM Transactions on Sensor Networks (TOSN)},
  vol.~11, no.~2, pp. 38:1--33, Feb 2015.

\bibitem{chang2012not}
C.-H. Chang, W.-C. Liao, H.-Y. Hsieh, and H.-J. Su, ``Not every bit counts:
  Shifting the focus from machine to data for machine-to-machine
  communications,'' in \emph{Proc.\ of Conference Record of the Forty Sixth
  Asilomar Conference on Signals, Systems and Computers}, Pacific Grove, CA,
  USA, Nov 2012, pp. 581--585.

\bibitem{hamidouche2017popular}
K.~Hamidouche, W.~Saad, and M.~Debbah, ``Popular matching games for
  correlation-aware resource allocation in the internet of things,'' in
  \emph{Proc.\ of IEEE Global Communications Conference (GLOBECOM)}, Singapore,
  Dec 2017, pp. 1--8.

\bibitem{pattem2008impact}
S.~Pattem, B.~Krishnamachari, and R.~Govindan, ``The impact of spatial
  correlation on routing with compression in wireless sensor networks,''
  \emph{ACM Transactions on Sensor Networks (TOSN)}, vol.~4, no.~4, pp.
  24:1--33, Aug 2008.

\bibitem{khan2011evolutionary}
M.~A. Khan and H.~Tembine, ``Evolutionary coalitional games in network
  selection,'' in \emph{Proc.\ of Wireless Advanced (WiAd)}, London, United
  Kingdom, June 2011, pp. 185--194.

\bibitem{luo2013evolutionary}
X.~Luo and H.~Tembine, ``Evolutionary coalitional games for random access
  control,'' in \emph{Proc.\ of IEEE International Conference on Computer
  Communications (INFOCOM)}, Turin, Itally, Apr 2013, pp. 535--539.

\bibitem{platkowski2016evolutionary}
T.~P{\l}atkowski, ``Evolutionary coalitional games,'' \emph{Dynamic Games and
  Applications}, vol.~6, no.~3, pp. 396--408, Sep 2016.

\bibitem{haenggi2012stochastic}
M.~Haenggi, \emph{Stochastic geometry for wireless networks}.\hskip 1em plus
  0.5em minus 0.4em\relax Cambridge, UK: Cambridge University Press, 2012.

\bibitem{chiu2013stochastic}
S.~N. Chiu, D.~Stoyan, W.~S. Kendall, and J.~Mecke, \emph{Stochastic geometry
  and its applications}.\hskip 1em plus 0.5em minus 0.4em\relax Hoboken, NJ,
  USA: John Wiley \& Sons, 2013.

\bibitem{jiang2014evolutionary}
T.~Jiang, X.~Tan, X.~Luan, X.~Zhang, and J.~Wu, ``Evolutionary game based
  access class barring for machine-to-machine communications,'' in \emph{Proc.\
  of 16th International Conference on Advanced Communication Technology
  (ICACT)}, Pyeongchang, South Korea, Feb 2014, pp. 832--835.

\bibitem{tembine2010evolutionary}
H.~Tembine, E.~Altman, R.~El-Azouzi, and Y.~Hayel, ``Evolutionary games in
  wireless networks,'' \emph{IEEE Trans. on Systems, Man, and Cybernetics, Part
  B (Cybernetics)}, vol.~40, no.~3, pp. 634--646, June 2010.

\bibitem{young2014handbook}
P.~Young and S.~Zamir, \emph{Handbook of game theory}.\hskip 1em plus 0.5em
  minus 0.4em\relax Elsevier, 2014, vol.~4.

\bibitem{sartakhti2017mmp}
J.~S. Sartakhti, M.~H. Manshaei, and M.~Sadeghi, ``{MMP}--{TIMP} interactions
  in cancer invasion: An evolutionary game-theoretical framework,''
  \emph{Journal of theoretical biology}, vol. 412, pp. 17--26, Jan 2017.

\bibitem{srinivasa2010distance}
S.~Srinivasa and M.~Haenggi, ``Distance distributions in finite uniformly
  random networks: Theory and applications,'' \emph{IEEE Transactions on
  Vehicular Technology}, vol.~59, no.~2, pp. 940--949, Feb 2010.

\bibitem{niyato2009dynamics}
D.~Niyato and E.~Hossain, ``Dynamics of network selection in heterogeneous
  wireless networks: An evolutionary game approach,'' \emph{IEEE Trans. on
  Vehicular Technology}, vol.~58, no.~4, pp. 2008--2017, May 2009.

\bibitem{yan2017evolutionary}
S.~Yan, M.~Peng, M.~A. Abana, and W.~Wang, ``An evolutionary game for user
  access mode selection in fog radio access networks,'' \emph{IEEE Access},
  vol.~5, pp. 2200--2210, Jan 2017.

\bibitem{shoham2008multiagent}
Y.~Shoham and K.~Leyton-Brown, \emph{Multiagent systems: Algorithmic,
  game-theoretic, and logical foundations}.\hskip 1em plus 0.5em minus
  0.4em\relax Cambridge, UK: Cambridge University Press, 2008.

\end{thebibliography}
\end{document}